\documentclass[a4paper, UKenglish, cleveref, thm-restate]{lipics-v2021}

\hideLIPIcs  

\usepackage{tabularx, environ}

\makeatletter

\newcolumntype{\expand}{}
\long\@namedef{NC@rewrite@\string\expand}{\expandafter\NC@find}

\NewEnviron{problem}[2][]{%
  \def\problem@arg{#1}%
  \def\problem@framed{framed}%
  \def\problem@lined{lined}%
  \def\problem@doublelined{doublelined}%
  \ifx\problem@arg\@empty%
    \def\problem@hline{}%
  \else%
    \ifx\problem@arg\problem@doublelined%
      \def\problem@hline{\hline\hline}%
    \else%
      \def\problem@hline{\hline}%
    \fi%
  \fi%
  \ifx\problem@arg\problem@framed%
    \def\problem@tablelayout{|>{\bfseries}lX|c}%
    \def\problem@title{\multicolumn{2}{|%
      >{\raisebox{-\fboxsep}}%
      p{\dimexpr\textwidth-4\fboxsep-2\arrayrulewidth\relax}%
      |}{%
        \textsc{\Large #2}%
      }}%
  \else
    \def\problem@tablelayout{>{\bfseries}lXc}%
    \def\problem@title{\multicolumn{2}{>%
      {\raisebox{-\fboxsep}}%
      p{\dimexpr\textwidth-4\fboxsep\relax}%
      }{%
        \textsc{\Large #2}%
      }}%
  \fi%
  \smallbreak\noindent%
  \begin{tabularx}{\textwidth}{\expand\problem@tablelayout}%
    \problem@hline%
    \problem@title\\%
    \BODY\\\problem@hline%
  \end{tabularx}%
  \smallbreak%
}
\makeatother

\title{Convolution and Knapsack in Higher Dimensions}

\author{Kilian Grage}{Kiel University, Germany}{kilian-g@t-online.de}{}{}

\author{Klaus Jansen}{Kiel University, Germany}{kj@informatik.uni-kiel.de}{https://orcid.org/0000-0001-8358-6796}{Supported by the German Research Foundation (DFG) project JA 612/25-1.}
\author{Björn Schumacher}{Kiel University, Germany}{bsch@informatik.uni-kiel.de}{https://orcid.org/0009-0002-2448-4794}{Supported by the German Research Foundation (DFG) project JA 612/25-1.}

\authorrunning{K. Grage, K. Jansen, and B. Schumacher} 

\Copyright{Kilian Grage and Klaus Jansen and Björn Schumacher} 



\nolinenumbers 

\EventEditors{Pat Morin and Eunjin Oh}
\EventNoEds{2}
\EventLongTitle{19th International Symposium on Algorithms and Data Structures (WADS 2025)}
\EventShortTitle{WADS 2025}
\EventAcronym{WADS}
\EventYear{2025}
\EventDate{August 11--15, 2025}
\EventLocation{York University, Toronto, Canada}
\EventLogo{}
\SeriesVolume{349}
\ArticleNo{41}

\usepackage{todonotes}

\usepackage{tikz}
\usetikzlibrary{patterns, shapes.multipart, shapes.geometric, fadings,
  arrows.meta, math, perspective, 3d, decorations.pathreplacing}

\usepackage{mathtools}
\usepackage{esvect}

\newcommand{\maxconv}[0]{{\textsc{$d$-MaxConv}}}
\newcommand{\maxconvup}[0]{{\textsc{$d$-MaxConv UpperBound}}}
\newcommand{\supadd}[0]{{\textsc{$d$-SuperAdditivity Testing}}}
\newcommand{\knaps}[0]{{\textsc{$d$-Knapsack}}}
\newcommand{\knapsone}[0]{{\textsc{0/1 $d$-Knapsack}}}
\newcommand{\knapsunb}[0]{{\textsc{Unbounded $d$-Knapsack}}}
\newcommand{\knapsbound}[0]{{\textsc{Bounded $d$-Knapsack}}}
\newcommand{\knapsoneeq}[0]{{\textsc{Bounded $d$-EqualityKnapsack}}}

\newcommand{\maxconvvar}[1]{{\textsc{$#1$-MaxConv}}}

\newcommand{\nullvec}[0]{{\constvec{0}}}
\newcommand{\constvec}[1]{\vv{\mathbf{#1}}}

\def\polylog{\operatorname{polylog}}

\DeclareMathOperator\poly{poly}

\DeclareDocumentCommand\tran{}{\ensuremath{^{\mathsf{T}}}}

\newcommand{\bZ}{\mathbb{Z}}
\newcommand{\bN}{\mathbb{N}}

\newcommand{\cO}{\mathcal{O}}
\renewcommand{\epsilon}{\varepsilon}

\DeclarePairedDelimiter\abs\lvert\rvert
\DeclarePairedDelimiter\norm\lVert\rVert
\DeclarePairedDelimiter\floor\lfloor\rfloor

\DeclarePairedDelimiter\paren{(}{)}

\DeclarePairedDelimiterXPP\bo[1]{\cO}{(}{)}{}{#1}

\DeclarePairedDelimiterX\setImpl[1]{\{}{\}}{#1}
\NewDocumentCommand\set{sO{}m>{\TrimSpaces}o}{
  \IfValueTF{#4}{
    \setImpl[#2]{#3 : \IfBooleanTF{#1}{\text{#4}}{#4}}
  }{
    \setImpl[#2]{#3}
  }
}

\allowdisplaybreaks

\begin{document}

\maketitle

\begin{abstract}
  In the Knapsack problem, one is given the task of packing a knapsack of a
  given size with items in order to gain a packing with a high profit value. As
  one of the most classical problems in computer science, research for this
  problem has gone a long way. One important connection to the
  $(\max,+)$-convolution problem has been established, where knapsack solutions
  can be combined by building the convolution of two sequences. This observation
  has been used in recent years to give conditional lower bounds but also
  parameterized algorithms.

  In this paper we carry these results into higher dimensions. We consider
  Knapsack where items are characterized by multiple properties -- given through
  a vector -- and a knapsack that has a capacity vector. The packing must not
  exceed any of the given capacity constraints. In order to show a similar
  sub-quadratic lower bound we consider a multidimensional version of
  \((\max, +)\)-convolution.  We then consider variants of this problem introduced by
  Cygan et al. and prove that they are all equivalent in terms of algorithms
  that allow for a running time sub-quadratic in the number of entries of the
  array.

  We further develop a parameterized algorithm to solve higher dimensional
  Knapsack. The techniques we apply are inspired by an algorithm introduced by
  Axiotis and Tzamos.  We will show that even for higher dimensional Knapsack,
  we can reduce the problem to convolution on one-dimensional, concave
  sequences, leading to an
  $\mathcal{O}(dn + dD \cdot \max\{\Pi_{i=1}^d{t_i}, t_{\max}\log t_{\max}\})$ algorithm, where
  $D$ is the number of different weight vectors, $t$ the capacity vector and $d$
  is the dimension of the problem.  Then, we use the techniques to improve the
  approach of Eisenbrand and Weismantel to obtain an algorithm for Integer
  Linear Programming with upper bounds with running time
  \(\bo{dn} + D \cdot \mathcal{O}(d \Delta)^{d(d+1)} + T_{\mathrm{LP}}\).

  Finally, we give an divide-and-conquer algorithm for ILP with running time
  \(n^{d+1} \cdot \mathcal O(\Delta)^d \cdot \log(\norm{u - \ell}_{\infty})\).
\end{abstract}
\section{Introduction}\label{sec:intro}
The Knapsack problem is one of the core problems in computer science. The task
of finding a collection of items that fits into a knapsack but also maximizes
some profit is NP-hard and as such the aim is to find approximation algorithms,
parameterized algorithms, and determine lower bounds for the running time of
exact algorithms.

Cygan et al.~\cite{cygan} and Künnemann et al.~\cite{Kunnemann} among others
have used the following relationship between Knapsack and the
$(\max,+)$-convolution problem. Assume we are given two disjoint item sets $A$
and $B$ and a knapsack of size $t$. If we additionally know the optimal profits
for all knapsack sizes $t' \le t$ we then can calculate the maximum profits for
all sizes $t'$ for $A\cup B$ by using convolution. This connection was used in
order to show that the existence of a sub-quadratic (in terms of capacity)
algorithm for Knapsack is equivalent to the existence of a sub-quadratic
algorithm for $(\max,+)$-convolution.

In this work, we consider these problems in higher dimensions. The Knapsack
problem is simply generalized by replacing the single value weight and capacity
by vectors. We then look for a collection of items whose summed up vectors do
not exceed the capacity vector in any position. The natural question arises
whether a similar quadratic time lower bound exists for this problem. We answer
this question positively by giving a generalization of convolution in higher
dimensions as well. Using this generalization and techniques introduced by
Bringmann \cite{bringmannSubsetSum} such as color-coding we are able to achieve
similar subquadratic lower bounds as in the one-dimensional case.

\subsection{Problem Definitions and Notations}
We define $[k] \coloneqq \set{i\in \mathbb{N}}[1 \le i \le k]$,
\([k]_0\coloneqq [k]\cup\set{0}\), and
\([a, b]\coloneqq \set{i\in\bZ}[a \leq i \leq b]\) for \(k, a, b\in\bZ\). In the following, we write
for two vectors $v,u\in \mathbb{R}^d$ that $v\le u$ (resp. $v<u$) if for all
$i \in [d]$ we have that $v_i\le u_i $ (resp. $v_i<u_i$). Further we denote with
$v_{\max} = \max_{i \in [d]}{v_i}$ for any vector $v\in \mathbb{R}^d$.  With
$\constvec{k} \in \mathbb{R}^d$ we denote the vector that has $k\in\mathbb{R}$ in every position.
\begin{definition}\label{def:array}
  Let $L \in \mathbb{N}^d$ be a $d$-dimensional vector and
  $A=(A_{i_1 i_2 \cdots i_d})$ be a $L_1\times L_2 \times \cdots \times L_d$ \emph{array}. We call the
  vector $L$ the \emph{size} of $A$ and denote the number of entries in $A$ with
  $\Pi(L) \coloneqq \prod_{i=1}^{d}{L_i}$.

  We call a vector $v\in \mathbb{Z}$ with $\nullvec{} \le v \le L-\constvec{1}$
  \textit{position} of array $A$. For ease of notation, we will denote for any
  array position $v$ of $A$ the respective array entry
  $A_{v_1 v_2 \cdots v_d}$ with $A_v$.
\end{definition}
We note that in this definition, array positions lie in between $\nullvec{}$
and $L - \constvec{1}$. This makes it easier to work with positions for
convolution and also for formulating time complexity bounds, as $\Pi(L)$ will be
the main parameter we consider.
\begin{definition}[Maximum Convolution]
  Given two $d$-dimensional arrays $A,B$ with equal size $L$, the
  $(\max,+)$-convolution of $A$ and $B$ denoted as $A\oplus B$ is defined as an array
  $C$ of size $L$ with $C_v \coloneqq \max_{u \le v}{A_u + B_{v-u}}$ for any $v< L$.
\end{definition}
\begin{problem}[framed]{\maxconv{}}
  Input: & Two $d$-dimensional arrays $A,B$ with equal size $L$. \\
  Problem: & Compute the array  $C \coloneqq A \oplus B$ of size $L$.
\end{problem}
\noindent Note that in the following we will refer to this problem as
``Convolution''. We specifically limit ourselves to the special truncated case
where both input arrays and the output array have the same size.  In a more
general setting, we could allow arrays of sizes $L^{(1)},L^{(2)}$ as input and
compute an array of size $L^{(1)}+L^{(2)}- \constvec{1}$.

To measure the running time of our algorithms, we mainly consider the size or
rather the number of entries from the resulting array because we need to
calculate a value for every position. Therefore, in terms of theoretical
performance, working with different sizes or calculating an array of combined
size will not make a difference. By only considering arrays of the same size,
we avoid many unnecessary cases and the dummy values of $-\infty$.

There is a quadratic time algorithm for \maxconv{} like in the one-dimensional
case. This algorithm simply iterates through all pairs of positions in time
$\mathcal{O}(\Pi(L)^{2})\subseteq \mathcal{O}(L_{\max}^{2d})$.

Further we define an upper bound test for the convolution problem. In this
problem, we are given a third input array and need to decide whether its
entries are upper bounds for the entries of the convolution.

\begin{problem}[framed]{\maxconvup{}}
  Input: & Three $d$-dimensional arrays $A,B,C$ with equal size $L$. \\
  Problem: & Decide whether $(A \oplus B)_v \le C_v$ for all $v < L$.
\end{problem}
\noindent We further generalize the notion of superadditivity to
multidimensional arrays.
\begin{problem}[framed]{\supadd{}}
  Input: & One $d$-dimensional array $A$ of size $L$. \\
  Problem: & Decide whether $A$ is superadditive, i.e., \(A_v\geq (A\oplus A)_v\) for
  all \(v < L\).
\end{problem}
\noindent The next problem class we consider is the \knaps{} problem. In the
one-dimensional case, one is given a set of items with weights and profits and
one knapsack with a certain weight capacity. The goal is now to pack the
knapsack such that the profit is maximized and the total weight of packed items
does not exceed the weight capacity.

The natural higher dimensional generalization arises when we have more
constraints to fulfill.  When going on a journey by plane, one for example has
several further requirements such as a maximum amount of allowed liquid or
number of suitcases. By imposing more similar requirements, we can simply
identify each item by a vector and also define the knapsack by a capacity
vector. This leads to the following generalization of Knapsack into higher
dimensions. We further differentiate between two problems \knapsone{} and
\knapsunb{}, depending on whether we allow items to be only used one time or an
arbitrary number of times.

\begin{problem}[framed]{\knapsone{}}
  Input: & Set $I$ of $n$ items each defined by a profit $p_i\in \mathbb{R}_{\ge 0}$ and weight vector $w^{(i)} \in \mathbb{N}^d$, along with a knapsack of capacity $t\in\mathbb{N}^d$. \\
  Problem: & Find subset $S\subseteq I$ such that $\sum_{i\in S}{w^{(i)}} \le t$ and $\sum_{i\in S}{p_i}$ is maximal.
\end{problem}

\begin{problem}[framed]{\knapsunb{}}
  Input: & Set $I$ of $n$ items each defined by a profit $p_i\in \mathbb{R}_{\ge 0}$ and weight vector $w^{(i)} \in \mathbb{N}^d$, along with a knapsack of capacity $t\in\mathbb{N}^d$. \\
  Problem: & Find a multi-set $S\subseteq I$ such that $\sum_{i\in S}{w^{(i)}} \le t$ and $\sum_{i\in S}{p_i}$ is maximal.
\end{problem}
\noindent The running time for these problems is mainly dependent on the dimension $d$, the number of items $n$ of an instance and the number of feasible capacities $\Pi(t+ \constvec{1})$ and we will further study the connection of these in regards to the convolution problems.

\subsection{Related Work}\label{hyp:maxconv}
Cygan et al.~\cite{cygan} as well as Künnemann et al.~\cite{Kunnemann} initiated the research and were the first to introduce this class of problems and convolution hardness. In particular, Cygan et al. showed for $d=1$ that all these problems are equivalent in terms of whether they allow for a subquadratic algorithm. This allows to formulate a conditional lower bound for all these problems under the hypothesis that no subquadratic algorithm for $\maxconv{}$ exists.

\begin{conjecture}[MaxConv-hypothesis]
There exists no $\mathcal O(\Pi(L)^{2-\epsilon})$ time algorithm for any $\epsilon >0$ for \maxconv{} with $d=1$.
\end{conjecture}
\subparagraph*{\((\max, +)\)-Convolution} The best known algorithm to solve
convolution on sequences without any further assumption takes time
$n^2/2^{\Omega(\sqrt{\log n})}$. This result was achieved by Williams
\cite{williams}, who gave an algorithm for APSP in conjunction with a reduction
by Bremner et al.~\cite{bremner}. However, the existence of a truly subquadratic
algorithm remains open.

Research therefore has taken a focus on special cases of convolution where one or both input sequences is required to have certain structural properties such as monotony, concavity or linearity. Chan and Lewenstein \cite{ChanMono} gave a subquadratic $\mathcal{O}(n^{1.864})$ algorithm for instances where both sequences are monotone increasing and the values are bound by $\mathcal{O}(n)$. Chi et al.~\cite{Chi} improved this further with a randomized $\widetilde{\mathcal{O}}(n^{1.5})$ algorithm.

Axiotis and Tzamos~\cite{axiotis19} showed that Convolution with only one
concave sequence can be solved in linear time.  Gribanov et al.~\cite{gribanov}
studied multiple cases and gave subquadratic algorithms when one sequence is
convex, piece-wise linear or polynomial.

\subparagraph*{Knapsack} Convolution has proven to be a useful tool to
solve other problems as well, in particular the Knapsack problem. In fact one of
the reductions of Cygan et al.~\cite{cygan} was an adapted version of
Bringmann's algorithm for subset sum \cite{bringmannSubsetSum}. Bringmann's
algorithm works by constructing sub-instances, solving these and then combining
the solutions via Fast Fourier Transformation (FFT). The algorithm by Cygan et
al. to solve Knapsack works the same way, but uses Convolution instead of FFT.

Axiotis and Tzamos follow a similar approach but choose their sub-instances more carefully, by grouping items with the same weight. That way, the solutions make up concave profit sequences which can be combined in linear time \cite{axiotis19}. This yields in total an $\mathcal O (n+Dt)$ algorithm, where $D$ is the number of different item weights.

Polak et al.~\cite{polak21} gave an
$\mathcal O (n+ \min\{w_{\max},p_{\max}\}^3)$ algorithm, where
$w_{\max},p_{\max}$ denote the maximum weight and profit respectively. They
achieved this by combining the techniques from Axiotis and Tzamos with proximity
results from Eisenbrand and Weismantel \cite{eisenbrand}. Further, Chen et
al.~\cite{chen2023faster} improved this to a time of
$\widetilde{\mathcal{O}}(n+w_{\max}^{2.4})$.  Recently, Ce Jin \cite{Jin23} gave an
improved $\widetilde{\mathcal{O}}(n+w_{\max}^{2})$ algorithm. Independently to these
results for 0/1 Knapsack, Bringmann gave an
$\widetilde{\mathcal{O}}(n+w_{\max}^{2})$ algorithm for Bounded Knapsack.

Doron-Arad et al.~\cite{ddimknapsacklb} showed that there are constants $\zeta, d_0
> 0$, such that for every integer \(d > d_0\) there is no algorithm that solves
\knapsone{} in time \(\bo{\paren{n + t_{\max}}^{\zeta \frac{d}{\log d}}}\).

\subparagraph*{Integer Linear Program} \(d\)-dimensional Knapsack problems are a special case of
integer linear programs.
\begin{problem}[framed]{Integer Linear Program (ILP)}
  Input: & A matrix $A\in \mathbb{Z}^{d\times n}$, a target vector $b\in \mathbb{Z}^d$, a profit vector $c\in \mathbb{Z}^n$ and an upper bound vector $u\in \mathbb{N}^d$.\\
  Problem: & Find $x\in \mathbb{Z}^n$ with $Ax = b$, $0 \le x \le u$ and such that $c\tran{}x$ is maximal.
\end{problem}
\noindent Lower bounds \(\ell\in \bN^d\) may be present but can easily be removed by
changing the right side to \(b - A\ell\) and the upper bound to \(u - \ell\).  If we
limit $A$ to be a matrix with only non-negative entries and
$u \coloneqq \constvec{1}$ then solving the above defined ILP is equivalent to
$\knapsone{}$. If we omit the $x\le u$ constraint, the resulting problem is
$\knapsunb{}$.

A very important part in research of ILPs has been on \textit{proximity}. The
proximity of an ILP is the distance between an optimal solution and an optimal
solution of its relaxation.  The relaxation of an ILP is given by allowing the
solution vector $x$ to be fractional.

Eisenbrand and Weismantel~\cite{eisenbrand} have proven that for an optimal
solution of the LP-relaxation $x^*\in \mathbb{R}^n$ there is an optimal integral solution
$z^* \in \mathbb{Z}^n$ with $||x^* -z^*||_1 \le d(2d\Delta+1)^d$ with $\Delta$ being the largest
absolute entry in $A$. They used this result to give an algorithm that solves
ILPs without upper bounds in time
$(d\Delta)^{\mathcal{O}(m)} \cdot \norm{b}_{\infty}^2$ and ILPs as defined above in time
$\mathcal{O}(n\cdot \mathcal{O}(d)^{(d+1)^2} \cdot \mathcal{O}(\Delta)^{d(d+1)} \cdot \log^2{(d\Delta)})$.  They additionally
extended their result to Bounded Knapsack and obtained the running time
$\mathcal{O}(n^2\cdot \Delta^2)$.

There have been further results on proximity such as Lee et
al.~\cite{proxspars}, who gave a proximity bound of
$3d^2\log{(2\sqrt{d}\cdot \Delta_m^{1/m})}\cdot \Delta_m $ where $\Delta_m$ is the largest absolute
value of a minor of order $m$ of matrix $A$. Celaya et al.~\cite{proxcel} also
gave proximity bounds for certain modular matrices.

Rohwedder and W\k{e}grzycki~\cite{larsdistinctalg} conjecture that there is no
\(2^{\bo{m^{2-\epsilon}}}\poly(n)\) time algorithm for ILP with
\(\Delta=\bo{1}\).  They also show that there are several problems that are
equivalent with respect to this conjecture.

\subsection{Our Results}

We begin by expanding the results of Cygan et al.~\cite{cygan} into higher
dimensions. The natural question is whether similar relations shown in their
work also exist in higher dimensions and in fact they do. In the first part of
this paper we show that the same equivalence -- regarding existing
subquadratic time algorithms -- holds among the higher dimensional problems.

\begin{theorem}\label{theo:cycle}
For any fixed $d$ and any $\epsilon > 0$, the following statements are equivalent:
\begin{enumerate}
    \item There exists an $\mathcal{O}(\Pi(L)^{2-\epsilon})$-time algorithm for \maxconv.
    \item There exists an $\mathcal{O}(\Pi(L)^{2-\epsilon})$-time algorithm for \maxconvup.
    \item There exists an $\mathcal{O}(\Pi(L)^{2-\epsilon})$-time algorithm for \supadd.
    \item There exists an $\mathcal{O}(n+\Pi(t)^{2-\epsilon})$-time algorithm for \knapsunb.
    \item There exists an $\mathcal{O}(n+\Pi(t)^{2-\epsilon})$-time algorithm for \knapsone.
\end{enumerate}
\end{theorem}
Some of these reduction incur a multiplicative factor of $\mathcal{O}(2^d)$. This is a
natural consequence due to the exponentially larger amount of entries that our
arrays hold and that we need to process. As an example, where it was
sufficient in the one-dimensional case to split a problem in two sub-problems, we
may now need to consider $2^d$ sub-problems. For this reason we require $d$ to
be fixed, so we can omit these factors.  We will prove this statement through a
ring of reductions. We note that one of the reductions uses an algorithm or a
generalization of it from Bringmann \cite{bringmannSubsetSum, cygan} that is
randomized. Part of this algorithm, involving so-called color-coding can be
derandomized, but a full derandomization is still an open problem.

We note that under the MaxConv-hypothesis, there does not exist an
$\mathcal{O}(\Pi(L)^{2-\epsilon})$-time algorithm for \maxconvvar{1}. If there exists any
$d$ such that \maxconv{} admits a sub-quadratic algorithm, then it would also
solve the problem in sub-quadratic time for any $d' \le d$ and especially
\(d' = 1\), hence contradicting the MaxConv-hypothesis, because we can extend
any $d'$-dimensional array to a $d$-dimensional one by adding dimensions with
size one.

We also discuss how to use lower order improvements for 1-dimensional
convolution for the \(d\)-dimensional convolution via a standard argument using
linearization of the vector and adding appropriate padding.  As this is a
standard trick we omit the proof in the main body but we state it in
\Cref{sec:calc-d-conv-helix} for completeness.
\begin{lemma}[restate=calcconvhelix,name=]\label{prop:calc-conv-helix}
  Given an algorithm for 1-dimensional convolution with running time \(T(n)\)
  and two \(d\)-dimensional arrays \(A, B\) with size \(L\), we can calculate
  \(A\oplus B\) in time \(T(2^d\prod(L))\).
\end{lemma}
Together with the discussion above, \Cref{prop:calc-conv-helix} shows that
1-dimensional and \(d\)-dimensional convolution are equivalent for fixed \(d\)
or up to a factor \(2^{\bo{d}}\).

In the second part of our paper, we complement our conditional lower bound with
a parameterized algorithm. To achieve this, we generalize an algorithm by
Axiotis and Tzamos~\cite{axiotis19}. Our algorithm will also have a running time
dependent on the number of different item weights, that we denote with $D$ and
the largest item weight $\Delta \coloneqq \max_{i\in I}\norm{w^{(i)}}_\infty$.  This algorithm will
also group items by weight vector and solve the respective sub-instances. The
resulting partial solutions are then combined via $\maxconv{}$.  However,
solving general $\maxconv{}$ needs quadratic time in the number of entries. We can
overcome this barrier by reducing our problem to convolution on one-dimensional,
concave sequences.

\begin{theorem}[restate=boundedKnapsackAlg,name=]\label{theo:algo}
  There is an algorithm for $\knapsoneeq{}$ with running time
  $\mathcal{O}(dn + d\cdot D \cdot \max\{\prod(t + \constvec{1}), t_{\max}\log t_{\max}\}) \subseteq \mathcal{O}(dn +
  d\cdot \min\set{n, (\Delta +1)^d} \cdot \max\{\prod(t + \constvec{1}), t_{\max}\log t_{\max}\}
  ))$.
\end{theorem}
In the general case our algorithm will achieve a running time of
$\mathcal{O}(dn + D \cdot \Pi(t)) $ as the $t_{\max}\log t_{\max}$ part only becomes relevant
when we have a slim knapsack, that is very large in one component, but
comparatively small in the other.
We note that our algorithm achieves the lower bound proposed
in \Cref{theo:cycle} since $D$ is also upper bounded by $\Pi(t)$.

We use the techniques for \Cref{theo:algo} to improve the approach of Eisenbrand
and Weismantel~\cite{eisenbrand} to obtain the following running time for ILP.
\begin{theorem}[restate=ilpaxiotis,name=]\label{prop:ilp-axiotis}
  There is an algorithm for ILP with running time
  \(\bo{dn} + D \cdot \mathcal{O}(d \Delta)^{d(d+1)} + T_{\mathrm{LP}}\).
\end{theorem}
Our approach reduces the dependency on the number of dimensions
(\(\mathcal{O}(d)^{d(d+1)}\) instead of \(\mathcal{O}(d)^{(d+1)^{2}}\)) and removes the logarithmic
factor \(\log^2(d\Delta)\).  Further, our algorithm is mostly independent on $n$ but
relies on the number of different columns of the given ILP.

In addition to their conjecture Rohwedder and W\k{e}grzycki~\cite{larsdistinctalg}
showed that the quadratic dependence on the number of constraints in the
exponent can be avoided if the term \(n^d\) is allowed in the runtime.  We
improve their algorithm by removing \(d\) from the base.
\begin{theorem}[restate=divConqAlg,name=]\label{prop:div-conq-alg}
  There is an algorithm that can solve ILP in time
  \(n^{d+1} \cdot \mathcal O(\Delta)^d \cdot \log(\norm{u - \ell}_{\infty})\).
\end{theorem}
The algorithm is a divide and conquer algorithm which repeatedly halves the
upper bounds.
\subsection{Organization of the Paper}
In \Cref{sec:redus} we present an overview of the reductions to proof
\Cref{theo:cycle}.  However, the concrete proofs are in
\Cref{sec:omitted-redus} as the proofs are similar to the ones by Cygan et
al.~\cite{cygan}.  Next we give the parameterized algorithm as well as the
generalization to ILPs in \Cref{sec:para-alg}.  The proof of the
divide-and-conquer algorithm for \Cref{prop:div-conq-alg} is given in
\Cref{sec:div-conq}.

\section{Reductions}\label{sec:redus}
For the Convolution problems, we will formulate the running time via
$T(\Pi(L),d)$, where $d$ is the dimension and $L$ is the size of the result array
- meaning the first parameter resembles the number of entries in the array. For
the Knapsack problems, we will add the number of items $n$ as parameter and
denote the running time of an algorithm via $T(n,\Pi(t+\constvec{1}),d)$, again
using $\Pi(t+\constvec{1})$ to denote the number of possible knapsack
capacities.

Similar to Cygan et al.~\cite{cygan}, we mainly look at functions satisfying
$T(\Pi(L),d) = c^{\mathcal{O}(d)} \cdot \Pi(L)^{\alpha}$ for some constants
$c,\alpha > 0$. Therefore, we remark that for all constant $c'>1$ we can write
$T(\Pi(c' \cdot L),d) \in \mathcal{O}(T(\Pi(L),d))$ since $d$ is fixed and we then have
$\Pi(c' L)^\alpha = c'^{d \cdot \alpha} \cdot \Pi(L)^\alpha$.

For all the mentioned problems we assume that all inputs, that is also any value
in any vector, consist of integers in the range of $[-W,W]$ for some
$W\in \mathbb{Z}$. This $W$ is generally omitted as a running time parameter and
$\polylog (W) $ are omitted or hidden in $T$ functions. We remark that --
unavoidably -- during the reductions we may have to deal with larger values than
$W$. We generally will multiply a factor of $\polylog (\lambda)$ in cases where the
values we have to handle increase up to $\lambda W$. Note that if $\lambda$ is a constant,
we again omit it.

We follow the same order as Cygan et al.~\cite{cygan}, because that makes the reduction
increasingly more complex. For the first reduction from \knapsunb{} to
\knapsone, we use the same reduction as in the one-dimensional case. The idea
is to simply encode taking a multitude of items via binary
encoding.

\begin{theorem}[restate=redbinenc,name={\knapsunb{} $\rightarrow$ \knapsone{}}]\label{theo:red1:binenc}
  A $T(n,\Pi(t),d)$ algorithm for \knapsone{} implies an
  $\mathcal{O}(T(n\log (t_{\max}), \Pi(t) ,d))$ algorithm for \knapsunb, where
  $t\in \mathbb{N}^d$ is the capacity of the knapsack.
\end{theorem}
For the next reduction we reduce \supadd{} to a special case where each array
entry is non-negative and values fulfill a monotony property.

\begin{theorem}[restate=redprimdual,name={\supadd{} $\rightarrow$ \knapsunb{}}]\label{theo:red:primdual}
  A \(T(n,\Pi(t),\allowbreak d)\) algorithm for \knapsunb{} implies the existence
  an algorithm that solves \supadd{} in time
  $\mathcal{O}(T(2\Pi(L),\Pi(2L),d) \polylog{(d L_{\max})})$ for an array $A$ with size $L$.
\end{theorem}
The next reduction differs from Cygan et al.~\cite{cygan}. We also combine our
input arrays together, but in the context of arrays we need to handle a number
of different combinations more than in the one-dimensional case. We therefore
add a block of negative values in an initial dummy block. This way, any
combination that is not relevant to the actual upper bound test, will result
positively when tested in the upper bound test.

\begin{theorem}[restate=redblockmatrix,name={\maxconvup{} $\rightarrow$ \supadd{}}]\label{theo:red:blockmatrix}
  A \(T(\Pi(L),\allowbreak d)\) algorithm for \supadd\ implies an $\mathcal{O}(T(4\Pi(L),d))$ algorithm for
  \maxconvup.
\end{theorem}
For the next reduction from \maxconv{} to \maxconvup, we will prove that we can
use an algorithm for \maxconvup{} to also identify a position that violates the
upper bound property.

\begin{theorem}[restate=redsquarestuff,name={\maxconv{} $\rightarrow$ \maxconvup{}}]\label{theo:red:squarestuff}
A $T(\Pi(L),d)$ algorithm for \maxconvup{} implies an $\mathcal{O}(2^d \Pi(L) \cdot T(2^{d}\sqrt{\Pi(L)},d) \cdot d \cdot \log(L_{\max}))$ algorithm for \maxconv{}.
\end{theorem}
The last reduction is based on Cygan et al.~\cite{cygan} and Bringmann
\cite{bringmannSubsetSum}. To make their approach work even in higher
dimensional cases, we require a new more refined distribution of items into
layers. With this new partition of items many used techniques such as
color-coding naturally translate into higher
dimensions.

\begin{theorem}[\knapsone{} $\rightarrow$ \maxconv{}]\label{theo:red:bringmann}
If \maxconv{} can be solved in time $T(\Pi(L),d)$ then \knapsone{} can be solved with a probability of at least $1-\delta$ in time $\mathcal{O}(T(\Pi(12t),d)\log (d\cdot t_{\max})\log^3{(\log n/\delta)}\cdot  d\cdot \log n)$ for any $\delta \in (0,1)$.
\end{theorem}
We remark that this also yields a respective algorithm for \knapsone{} with the
$\Pi(L)^2$ algorithm for \maxconv{}.  Lower order
improvements like the results of Bremner \cite{bremner} or Chan and Lewenstein
\cite{ChanMono} can be applied by calculating the convolution
using \Cref{prop:calc-conv-helix}.

\section{Parameterized Algorithm for Multi-Dimensional Knapsack}\label{sec:para-alg}
In this part, we will consider solving Knapsack such that the capacity is
completely utilized.  This enables a concise presentation and is not a
restriction as discussed after the concrete problem definition.

\begin{problem}[framed]{\knapsoneeq{}}
  Input: & Set $I$ of $n$ items each defined by a profit $p_i\in \mathbb{R}_{\ge 0}$, weight
  vector $w^{(i)} \in \mathbb{N}^d$, and upper bound \(u_i\in\bN\) along with a knapsack of capacity $t\in\mathbb{N}^d$. \\
  Problem: & Find a vector \(x\in\bN^n\) such that \(x\leq u\), $\sum_{i=1}^n{x_iw^{(i)}} = t$, and $\sum_{i=1}^n{x_ip_i}$ is maximal.
\end{problem}
\noindent We will not only solve this problem for the given capacity $t$ but for all
smaller capacities.  In detail, we will construct a solution array $A$ with size
$t+\constvec{1}$ such that for any $v\le t$ we have that $A_v$ is the maximum
profit of an item set whose total weight is exactly $v$.  Note that we can
construct a solution for \knapsbound{} after solving \knapsoneeq{} by
remembering the highest profit achieved in any position.

The algorithm we will show is based on the algorithm of Axiotis and Tzamos
\cite{axiotis19}. They solved one-dimensional Knapsack by splitting all items
into subsets with equal weight. They proceed then to solve these resulting
sub-instances. These sub-instances can be easily solved by gathering the highest
profit items that fit into the knapsack. Finally, they combine these resulting
partial solutions via convolution. The profits of one such partial solution
forms a concave sequence, which allows each convolution to be calculated in
linear time.

\begin{definition}[Concave Sequences]
  Let $a=(a_i)_{i \le n}$ be a sequence of numbers of length $n$. We call the
  sequence $a$ \emph{concave} if for all $i \le n-2$ we have
  $a_{i+1} - a_i \ge a_{i+2} - a_{i+1}$.
\end{definition}
\begin{lemma}[{\cite[Lemma 9]{axiotis19}}]\label{prop:linear-concave-conv}
  Given an arbitrary sequence of length \(m\) and a concave sequence of length
  \(n\) we can compute the \((\max, +)\) convolution of \(a\) and \(b\) in
  time \(\bo{m + n}\).
\end{lemma}
We achieve a similar result and calculate higher dimensional convolutions in
linear time. In fact, we will reduce our problem to calculating convolution of
one-dimensional concave sequences. This way we can calculate the convolution of
our sub-solutions in linear time in the number of entries in our array.

\boundedKnapsackAlg*
\noindent We remark, that in most cases
$V \coloneqq \prod(t + \constvec{1})> t_{\max} \log t_{\max} $ so generally this algorithm
will have a running time of $\mathcal{O}(n+DV )$. However, in the special case of a slim
array that has size $t$ with one large value $t_j$ and other values being very
small $t_i \ll t_j$ for $i \neq j$ our algorithm might have a slightly worse running
time of $\mathcal{O}(n+D t_{\max} \log t_{\max})$.  Note that is algorithm also works for
\knaps\ by setting every upper bound to \(t_{\max}\).
\begin{proof}[Proof of \Cref{theo:algo}]
  Since we have \(D\) different weights, let \(w(1), \ldots, w(D)\) be the different
  weights.  We partition the items \(I\) into \(D\) sets of the same weight,
  \(I_{w(i)} \coloneqq \set{j\in I}[w(i) = w^{(j)}]\) for all \(i\in[D]\).  We start with
  an empty solution array \(R^0\) of size \(t + \constvec{1}\) such that
  \(R^0_{\constvec{0}} = 0\) and \(R^0_v = -\infty\) for \(v \ne \constvec{0}\).  We will
  calculate solution arrays \(R^i\) for \(i\in[D]\) where \(R^i\) is the
  solution array for instance restricted to the items
  \(\bigcup_{j = 1}^iI_{w(j)}\).  Next we describe how we can calculate
  \(R^i\) from \(R^{i - 1}\).

  Consider an optimal solution for some position \(v\) in \(R^i\), i.e.,
  containing only items from \(\bigcup_{j = 1}^iI_{w{(j)}}\).  The solution
  contains a certain number of items of weight \(w{(i)}\).  Let this number be
  \(k\).  We have \(v' \coloneqq v - kw{(i)} \ge \constvec{0}\).  Thus, the other items
  in the solution are from \(\bigcup_{j = 1}^{i - 1}I_{w{(j)}}\) and have combined
  profit \(R^{i - 1}_{v'}\) otherwise our initial solution is not optimal or
  \(R^{i - 1}\) is not correct.  Also, the \(k\) items with weight \(w{(i)}\) in
  the optimal solution have the highest profit possible, otherwise the solution
  is not optimal.  Therefore, we have
  \begin{equation}
    \label{eq:concave-conv-is-correct}
    R^i_v = \smashoperator{\max_{\substack{k \geq 0\\v - k \cdot w{(i)} \geq \constvec{0}}}} R^{i - 1}_{v - k \cdot w} + f_w(k)
  \end{equation}
  where \(f_{w{(i)}}(k)\) is the largest profit for taking \(k\) items of weight
  \(w{(i)}\).  Hence, only position pairs \(v, u\) such that their difference is
  an integer multiple of \(w{(i)}\) (\(v - u \in w{(i)}\bZ\)) are relevant for the
  calculation of \(R^i_v\) or \(R^i_u\).  This condition forms an equivalence
  relation and thus we can calculate the values of \(R^i\) per equivalence
  class.  Every equivalence class has the structure
  \(\set{v, v + w{(i)}, v + 2\cdot w{(i)}, v + k_v \cdot w{(i)}} \) with
  \(v - w{(i)} \not\geq \constvec{0}\).  Define the sequences
  \((r_j)_{j\leq k_v}\) with \(r_j \coloneqq R^{i - 1}_{v + j \cdot w{(i)}}\) and
  \((a_j)_{j\leq k_v}\) with \(a_j \coloneqq f_{w{(i)}}(j)\).  Thus, we can calculate
  \(R^i\) by \(R^i_{v + j\cdot w{(i)}} \coloneqq (r \oplus a)_j\) for every
  \(j\in[k_v]\).  This is sufficient by \Cref{eq:concave-conv-is-correct}.  Note
  that, we only have to calculate the convolution once as we can use the result
  for every element of the equivalence class.  We always have
  \(k_v \leq t_{\max}\) and \((a_j)_{j\le k_v}\) is a concave sequence.

  The items can be partitioned in the sets \(I_{w{(i)}}\) in time
  \(\bo{dn + \Delta}\).  In every set the \(t_{\max}\) most profitable items can be
  calculated in total time \(\bo{n}\) by using
  \cite{DBLP:journals/jcss/BlumFPRT73}.  Calculating all values for
  \(f_{w{(i)}}(j)\) with \(i\in[D]\) and \(j\in t_{\max}\) takes time
  \(\bo{Dt_{\max}\log(t_{\max})}\).  Every convolution can be calculated in
  linear time by \Cref{prop:linear-concave-conv}.  Since the equivalence classes
  partition the set of positions, we can calculate \(R^i\) from \(R^{i - 1}\) in
  time \(d\cdot \bo{\max\set{V, t_{\max}\log(t_{\max})}}\) where the \(d\) factor is
  due to multidimensional indices.  We may assume \(\Delta \leq t_{\max}\) which shows
  the claimed running time.
\end{proof}

\subsection{Generalisation to ILPs}\label{section:ILP}
As we mentioned before, Knapsack is a special case of ILP. We now want to
discuss how to extend our results to ILPs.  In the case that all matrix entries
are non-negative ILP is equivalent to \knapsoneeq.  We discuss later how to
handle negative entries in the constraint matrix.  Negative item profits are
sensible in \knapsoneeq\ as those may be required to reach the capacity.  The
algorithm presented in \Cref{theo:algo} also works with negative profits.  The
running time of \Cref{theo:algo} depends on the capacity of the knapsack.  To
bound this value for ILPs we use the proximity bound by Eisenbrand and
Weismantel~\cite{eisenbrand} and improve their algorithmic approach.  We start
by giving a short summary of their approach.

For an ILP Instance with constraint matrix \(A\in\bZ^{d\times n}\), right side
\(b\in \bZ^d\), and upper bounds \(u\in\bN^n\), let
\(\Delta \coloneqq \norm{A}_{\max}\) be the biggest absolute value of an entry in the
constraint matrix.  Eisenbrand and Weismantel \cite{eisenbrand} calculate an
optimal solution $x^*$ for the LP-relaxation, which is possible in polynomial
time.  They proved that there exists an integral optimal solution $z^*$, such
that the distance in $\ell_1$-norm is small, more precisely
$\norm{z^* - x^*}_1 \le d \cdot (2d\Delta+1)^d$.  They then proceed to take
$\lfloor x^*\rfloor$ as an integral solution
(\(\norm{z^{*} - \floor{x^{*}}}_1\leq d\cdot(2d\Delta + 1)^d + d \eqqcolon L\)) and
construct the following ILP to find an optimal solution.
\begin{align*}
  \max c\tran{}x \text{ subject to} \\
  Ax &= b - A\lfloor x^* \rfloor\\
  \norm{x}_1 &\leq L\\
  \ell' &\leq x \leq u'\\
  x&\in\bZ^n
\end{align*}
Note that we may have \(\ell'_i < 0\) for some \(i\in[n]\).  Combining an optimal
solution to this ILP with \(\floor{x^{*}}\) yields an optimal solution to the
original ILP.  For any \(x\in\bZ^n\) with \(\norm{x}_1\leq L\) we have
\(\norm{Ax}_{\infty} \leq \Delta L \leq \bo{d\Delta}^{d+1}\) due to the bound on the
\(\ell_1\)-norm.  Eisenbrand and Weismantel construct an acyclic directed graph and
find the longest path in that graph which is equivalent to an optimal solution
to the ILP above.  They bound the size of the graph using the bound on the
\(\ell_{\infty}\)-norm above.  We avoid such a graph construction for our algorithm
which allows us to improve the running time.  For two \(i, j\in[n]\) with
\(A_i = A_j\) and \((c_i, i) > (c_j, j)\) we may assume that \(x^{*}_j > 0\)
implies \(x^{*}_i = u_i\).  Otherwise, we can modify \(x^{*}\) to obtain a
solution that adheres to this property with at least the same value.  This
structure allows us to reformulate the ILP by grouping variables together with
the same column in \(A\) as follows
\begin{align*}
  \max \sum\nolimits_{i\in[D]} f_i(x_i) \text{ subject to}\\
  A'x&= b - A\lfloor x^* \rfloor\\
  \norm{x}_1 &\leq L\\
  \ell'' &\leq x \leq u''\\
  x &\in\bZ^D
\end{align*}
where \(f_i\) is a concave function for all \(i\in[D]\).  We calculate an optimal
solution to this ILP using the convolution approach from \Cref{theo:algo}.
\ilpaxiotis*
\begin{proof}
  For \(i\in[D]_0\) let \(U^i\) be the best possible values using the variables
  \(x_j\) for \(j \leq i\).  The size of \(U^i\) in every dimension is
  \(2 \Delta L + 1\).  To calculate \(U^i_v\) from \(U^{i - 1}\) for some position
  \(v\in\bZ^d\) such that \(\norm{v}_{\infty}\leq \Delta L\) we need to calculate
  \[
    U^i_v = \max_{\substack{\ell''_i \leq j \leq u''_i\\\norm{v - jA'_i}_{\infty}\leq \Delta L}}U^{i - 1}_{v -
      jA'_i} + f_i(j).
  \]
  As in the proof of \Cref{theo:algo} we can consider the equivalence classes.
  Let \(S\) be the equivalence class of \(v\), it has the structure
  \(\set{v', v' + A'_i, v' + 2A'_i, \ldots, v' + (\abs{S} - 1)A'_i}\).  Define the
  sequence \((a_j)_{j\leq2\abs{S} - 2}\) with
  \(a_j \coloneqq U^{i - 1}_{v' + jA'_i}\) for \(j < \abs{S}\) and
  \(a_j \coloneqq \infty\) for \(j \geq \abs{S}\).  Further we define the sequence
  \((b_j)_{j \leq 2\abs{S} - 2}\) with \(b_j \coloneqq f_i(j - \abs{S} + 1)\).  We calculate
  the convolution of the sequences \(c \coloneqq a \oplus b\) in time
  \(\bo{\abs{S}}\) using \Cref{prop:linear-concave-conv}.  Then we get
  \begin{align*}
    U^i_{v' + jA'_i}
    &= \max_{\substack{\ell''_i \leq k \leq u''_i\\\norm{v' + jA'_i - kA'_i}_{\infty}\leq \Delta L}}U^{i - 1}_{v' + jA'_i - kA'_i} + f_i(k)\\
    &= \max_{j - \abs{S} + 1 \leq k \leq j}U^{i - 1}_{v' + (j - k)A'_i} + f_i(k)\\
    &= \max_{j\leq k \leq j + \abs{S} - 1 }U^{i - 1}_{v' + (j - k + \abs{S} - 1)A'_i} + f_i(k - \abs{S} + 1)\\
    &= \max_{0\leq k \leq \abs{S} - 1 }U^{i - 1}_{v' + (j - k - j + \abs{S} - 1)A'_i} + f_i(k + j - \abs{S} + 1)\\
    &= \max_{0\leq k \leq \abs{S} - 1 }U^{i - 1}_{v' + (\abs{S} - 1 - k)A'_i} + f_i(j + k - \abs{S} + 1)\\
    &= \max_{0\leq k\leq \abs{S} - 1} U^{i-1}_{v' + kA'_i} + f_i(j - k)\\
    &= \max_{0\leq k\leq \abs{S} - 1} U^{i-1}_{v' + kA'_i} + f_i(j - k + \abs{S} - 1 - \abs{S} + 1)\\
    &= \max_{0\leq k\leq \abs{S} - 1} a_k + b_{j - k + \abs{S} - 1}\\
    &= \max_{0\leq k\leq j +\abs{S} - 1} a_k + b_{j - k + \abs{S} - 1}\tag{\(a_k = -\infty\) for \(k \geq \abs{S}\)}\\
    &= c_{j + \abs{S} - 1}
  \end{align*}
  for \(j\in[\abs{S} - 1]_0\).
  Next discuss how to obtain the running time of \(\bo{d\Delta}^{d(d+1)}\) instead of
  the expected \(d\bo{d\Delta}^{d(d+1)}\) as in the proof of \Cref{theo:algo}.  To do
  this we add buffer zones of size \(\Delta\) at the start and the end of every
  dimension.  Then the size for every dimension is still bounded by
  \(\bo{d\Delta}^{d + 1}\).  Next we save the \(d\)-dimensional array as a
  \(1\)-dimensional array by concatenating the dimensions.  We mark the entries
  in the buffer zones.  For a fixed column we can calculate a step size in the
  \(1\)-dimensional array that corresponds to a step with the column in the
  \(d\)-dimensional array.  Using the marked buffer zones we can check if we
  stepped outside the valid area which allows us to find the equivalence
  class of a position in linear time of the size of the class.  Therefore, we
  can find all equivalence classes in time \(\bo{d\Delta}^{d(d+1)}\), and we can
  calculate the convolution in the same time.

  First we need to solve the linear relaxation of the ILP.  Denote the necessary
  time by \(T_{\mathrm{LP}}\).  Next we partition the columns by weight in time
  \(\bo{dn} + \Delta\).  We can find the \(L\) most profitable picked and unpicked
  items for every column type in total time \(\bo{n}\) using
  \cite{DBLP:journals/jcss/BlumFPRT73}.  With this we can calculate all relevant
  values for every \(f_i\) in time \(D\cdot L\log L\).  As explained above every
  convolution can be calculated in time \(\bo{d\Delta}^{d(d+1)}\) which yields the
  claimed running time.  An optimal solution vector can be calculated by
  backtracking in time \(D \cdot L\leq D\cdot d\cdot\bo{d\Delta}^{d+1}\)
\end{proof}

\section{Divide and Conquer Algorithm}\label{sec:div-conq}
In this section we show \Cref{prop:div-conq-alg}.  We start by showing the
central structural property.
\begin{lemma}[restate=structurelemma,name=]\label{prop:half-with-lbs}
  Let \(A\in\bZ^{d\times n}\), \(b\in\bZ^d\), and \(x, u\in\bN^d\).  If
  \(Ax = b\) and \(x \leq u\), there exists \(x'\in\bN^d\) such that
  \[
    \norm{2Ax' - b}_{\infty} \leq 2n\Delta, \quad x' \leq u' \coloneqq \floor[\Big]{\frac{u -
        1}{2}}\quad\text{and}\quad 0\leq x_i - 2x'_i \leq u_i - 2u_i' \leq 2 \text{ for
    } i\in[n].
  \]
\end{lemma}
\begin{proof}
  Let \(0 \leq x \leq u\) with \(Ax = b\).  We define the new solution \(x'\)
  component wise as follows.
  \[
    x'_i \coloneqq
    \begin{cases}
      0 & \text{if } x_i = 0,\\
      \floor{\frac{x_i}{2}} & \text{if } u_i \text{ is odd},\\
      \floor{\frac{x_i - 1}{2}} & \text{otherwise}
    \end{cases}
  \]
  for \(i\in[n]\).  We start by showing \(x' \leq u'\).  Let \(i\in[n]\).  If
  \(u_i\) is odd, we have
  \[
    x'_i = \floor[\Big]{\frac{x_i}{2}} \leq \frac{x_i}{2} \leq \frac{u_i}{2} = \frac{2u_i' +
    1}{2} = u'_i + \frac{1}{2} \implies x'_i \leq u'_i
  \]
  and
  \[
    0\leq x_i - 2x'_i = x_i - 2\floor[\Big]{\frac{x_i}{2}} \leq x_i -
    2 \frac{x_i - 1}{2} = 1 = u_i - 2u_i' \leq 2.
  \]
  Otherwise, if \(u_i\) is even, we have
  \[
    x'_i = \floor[\Big]{\frac{x_i - 1}{2}} \leq \frac{x_i - 1}{2} \leq \frac{u_i - 1}{2} =
    \frac{2u_i' + 1}{2} = u'_i + \frac{1}{2} \implies x'_i \leq u'_i
  \]
  and
  \[
    0 \leq x_i - 2x'_i = x_i - 2\floor[\Big]{\frac{x_i - 1}{2}} \leq x_i -
    2 \frac{x_i - 2}{2} = 2 = u_i - 2u_i' \leq 2.
  \]
  Next, let \(b' \coloneqq Ax'\).  We have
  \[
    \abs{2b'_j - b_j}
    = \abs[\Big]{\sum_{i=1}^n A_{j, i}(2x'_i - x_i)}
    \leq \sum_{i=1}^n\abs{A_{j, i}(2x'_i - x_i)}
    = \sum_{i=1}^n\abs{A_{j, i}}\abs{2x'_i - x_i}
    \leq 2n\Delta
  \]
  for every \(j\in[d]\) and thus \(\norm{2b' - b}_{\infty} \leq 2n\Delta\).
\end{proof}
This lemma implies that \emph{every} solution \(x\in\bN^n\) to \(Ax = b\) with
\(x\leq u\) can be decomposed into \(x',x''\in\bN^n\) such that
\(x = 2 x' + x''\), \(x'\) is a solution to \(Ax = b'\) and
\(x'' \leq u - 2u'\).  Note that \(u'\) does not depend on the concrete solution
but is the same for all.

We can iterate \Cref{prop:half-with-lbs} to obtain a series of solutions and
upper bounds until we reach a point where all upper bounds are zero.  Define
\(x^{(j)}\in\bN^n\) and \(u^{(j)}\in\bN^n\) after applying \Cref{prop:half-with-lbs}
\(j\) times, i.e., \(x^{(0)} = x\) and \(u^{(0)} = u\).  Let \(k\in\bN\) be smallest
value such that \(u^{(k)} = 0\).
\begin{corollary}\label{prop:iterate-halving}
  We have
  \(\norm[\big]{Ax^{(j)} - \frac{b}{2^j}}_{\infty} \leq n\Delta\sum_{i=1}^j \frac{1}{2^{i - 1}} \leq
  2n\Delta \) for \(j\leq k\).
\end{corollary}
\begin{proof}
  By induction.  \(j = 1\) is \Cref{prop:half-with-lbs}.  Let \(j > 1\).
  Then, we have
  \begin{align*}
    \norm[\Big]{Ax^{(j)} - \frac{b}{2^j}}_{\infty}
    &= \norm[\Big]{Ax^{(j)} - \frac{Ax^{(j - 1)}}{2} + \frac{Ax^{(j - 1)}}{2} - \frac{b}{2^j}}_{\infty}\\
    &\leq \norm[\Big]{Ax^{(j)} - \frac{Ax^{(j - 1)}}{2}}_{\infty} +
      \frac{1}{2}\norm[\Big]{Ax^{(j - 1)} - \frac{b}{2^{j - 1}}}_{\infty} \\
    &\leq n\Delta + \frac{n\Delta}{2}\sum_{i=1}^{j - 1}\frac{1}{2^{i-1}}
    = n\Delta + n\Delta\sum_{i=2}^{j}\frac{1}{2^{i-1}}
    = n\Delta\sum_{i=1}^{j}\frac{1}{2^{i-1}}.\qedhere
  \end{align*}
\end{proof}
With this we can give the algorithm and proof \Cref{prop:div-conq-alg}.

\begin{figure}[tb]
  \centering
  \begin{tikzpicture}[ell/.style={ellipse, draw, minimum height=1cm, minimum
      width=.75cm}, node distance=.75cm, >={Stealth[round]}]
    \node[ell] (k0) {\(0\)};
    \node[ell, right=of k0] (k1) {\(1\)};
    \node[right=of k1] (kd) {\(\cdots\)};
    \node[ell, right=of kd] (kn) {\(n\)};
    \draw[decoration=brace, decorate] ([yshift=1.2cm]k0.west) -- node[above,
    yshift=.2cm] {\(k - 1\)} ([yshift=1.2cm]kn.east);

    \node[ell, right=1.5cm of kn] (kk0) {\(0\)};
    \node[ell, right=of kk0] (kk1) {\(1\)};
    \node[right=of kk1] (kkd) {\(\cdots\)};
    \node[ell, right=of kkd] (kkn) {\(n\)};
    \draw[decoration=brace, decorate] ([yshift=1.2cm]kk0.west) -- node[above,
    yshift=.2cm] {\(k - 2\)} ([yshift=1.2cm]kkn.east);

    \node[below=.3cm of kk0] (dds) {\(\cdots\)};
    \fill[path fading=west] ([yshift=.2pt]dds.east) rectangle ++(1, -.4pt);
    \draw ([xshift=1cm]dds.east) .. controls +(2, 0) and +(0, -.65)
    .. node[font=\small, below, pos=.25] {double} (kkn.south);
    \fill[path fading=east] ([yshift=.2pt]dds.west) rectangle ++(-1, -.4pt);

    \node[ell, above left=-3cm and 1.5cm of kk0] (00) {\(0\)};
    \node[ell, right=of 00] (01) {\(1\)};
    \node[right=of 01] (0d) {\(\cdots\)};
    \node[ell, right=of 0d] (0n) {\(n\)};
    \draw[decoration=brace, decorate] ([yshift=-.75cm]0n.east) -- node[below, yshift=-.2cm] {\(0\)} ([yshift=-.75cm]00.west);

    \draw[->] ([xshift=-1cm]dds.west) -- ++(-1.075, 0) .. controls +(-1, 0) and
    +(-1.5, 0) .. node[left, font=\small] {double} (00);

    \draw[->] (k0) -- node[above, font=\small, yshift=.5cm] {\(u_1^{(k - 1)} - 2u_1^{(k)}\)} (k1);
    \fill[path fading=east] ([yshift=.2pt]k1.east) rectangle ([yshift=-.2pt]kd.west);
    \fill[path fading=west] ([yshift=.2pt]kd.east) rectangle ([yshift=-.2pt]kn.west);
    \draw[->, shorten <=.65cm] (kd) -- (kn);
    \draw[->] (kn) -- node[below, sloped, font=\small] {double} (kk0);

    \draw[->] (kk0) -- node[above, font=\small, yshift=.5cm] {\(u_1^{(k - 2)} - 2u_1^{(k - 1)}\)} (kk1);
    \fill[path fading=east] ([yshift=.2pt]kk1.east) rectangle ([yshift=-.2pt]kkd.west);
    \fill[path fading=west] ([yshift=.2pt]kkd.east) rectangle ([yshift=-.2pt]kkn.west);
    \draw[->, shorten <=.65cm] (kkd) -- (kkn);

    \draw[->] (00) -- node[above, font=\small, yshift=.5cm] {\(u_1^{(0)} - 2u_1^{(1)}\)} (01);
    \fill[path fading=east] ([yshift=.2pt]01.east) rectangle ([yshift=-.2pt]0d.west);
    \fill[path fading=west] ([yshift=.2pt]0d.east) rectangle ([yshift=-.2pt]0n.west);
    \draw[->, shorten <=.65cm] (0d) -- (0n);
  \end{tikzpicture}
  \caption{Structure of the graph in the proof of \Cref{prop:div-conq-alg}.}\label{fig:ex:graph}
\end{figure}
\divConqAlg*
\begin{proof}
  We find an optimal solution to the bounded ILP by finding the longest path in
  a graph that is constructed based on \Cref{prop:half-with-lbs}.  We define the
  graph as follows.
  \[
    V \coloneqq \set[\Big]{(b', j, i)}[j\in[k]_0, \norm[\Big]{b' - \frac{b}{2^j}}_{\infty} \leq
    2(2n - i)\Delta, i\in[n]_0]
  \]
  The intuition here is that we have a layer for every application of
  \Cref{prop:half-with-lbs}.  We can jump from one layer to the next by doubling
  the corresponding right side (the first component) and length of the path to
  the current vertex.  Then, to reconstruct the solution prior to an application
  of \Cref{prop:half-with-lbs} we may need to increase variables by up to two.
  For this we have the third component.  While going to a vertex with
  \(i\in[n]\) in the last component we may use the \(i\)th column up to the number
  of times it was removed to get the new upper bound.  This is also an upper
  bound for the number of times it was removed from a solution as stated in
  \Cref{prop:half-with-lbs}.  More precisely, we associate the ILP
  \begin{align*}
    \max c\tran{}x & \\
    Ax &= b'\\
    x_{\ell} &\le u^{(j)}_{\ell} \tag{for \(\ell\in[i]\)}\\
    x_{\ell} &\leq 2u^{(j + 1)}_{\ell} \tag{for \(\ell > i\)}\\
    x&\in\bN^n
  \end{align*}
  with a vertex \((b', j, i)\in V\) and refer to it by
  \(\mathrm{ILP}_{(b', j, i)}\).  A path from \((0, k - 1, 0)\) to
  \((b', j, i)\) describes a solution to \(\mathrm{ILP}_{(b', j, i)}\) where
  the value is equal to the length of the path.

  Next we give the precise definitions for the edges in the graph.  For a vertex
  with \(j < k\) and \(i < n\) we define the following edges
  \[
    (b', j, i) \xrightarrow{xc_{i + 1}} (b' + xA_{i + 1}, j, i + 1)
  \]
  with \(x\in[u_{i + 1}^{(j)} - 2u^{(j + 1)}_{i + 1}]_0\subseteq[2]_0\) and for vertices with
  \(j > 0\) and \(i = n\) we define the  edge
  \[
    (b', j, n) \xrightarrow{\text{double the length}} (2b', j - 1, 0).
  \]
  Because of the changes in the second and third component the graph is acyclic,
  and thus the longest path can be calculated in linear time by utilizing a
  topological order of the graph.  Note that the second type of edge does not have
  a classic length, but this doubling of the current length works with the
  aforementioned algorithm.  The out degree of every vertex is bounded by \(3\)
  and thus the size of graph is linear in the number vertices which is
  \(n^{d + 1} \cdot \bo{\Delta}^d \cdot \log(\norm{u}_{\infty})\).  Now we find the longest path
  from \((0, k - 1, 0)\) to \((b, 0, n)\).  The only solution to
  \(\mathrm{ILP}_{(0, k-1, 0)}\) is \(\constvec{0}\).  The structure of the
  graph is depicted in \Cref{fig:ex:graph}.

  Let \(x\in\bN^n\) be a solution to the ILP.  Further, let
  \(x^{(j)}\in\bN^n\) be the solutions obtained by iterating
  \Cref{prop:half-with-lbs}.  Define
  \[
    x^{(j, i)} \coloneqq
    \begin{cases*}
     2x^{(j + 1)} & if \(i = 0\),\\
     x^{(j, i - 1)} + \paren[\big]{x^{(j)}_i - 2x^{(j + 1)}_i}e_i& otherwise
    \end{cases*}
  \]
  for \(j\in[k - 1]_0\) and \(i\in[n]_0\) where \(e_i\) is the \(i\)th unit vector.  We have
  \(\norm[\big]{Ax^{(j, i)} - \frac{b}{2^{j}}}_{\infty} \leq 2(2n - i)\Delta \)
  by \Cref{prop:half-with-lbs,prop:iterate-halving}.  Thus,
  \begin{gather*}
    (0 = Ax^{(k - 1, 0)}, k - 1, 0), (Ax^{(k - 1, 1)}, k - 1, 1), \ldots, \\
    \quad (Ax^{(k - 1, n)}, k - 1, n),(Ax^{(k - 2, 0)}, k - 2, 0), \ldots, (b = Ax^{(0, n)}, 0, n)
  \end{gather*}
  is a path in the graph with length \(c\tran{} x\).  Therefore, ILP has a
  solution if and only if there is a path in the graph from \((0, k - 1 , 0)\)
  to \((b, 0, n)\).  Also the value of the optimal solution to the ILP is equal
  to the length of the longest path in the graph.
\end{proof}

\section{Conclusion}
First, we have proven that the relationship between Knapsack and Convolution in
the one-dimensional case also prevails in higher dimensional variants.  A natural
follow-up question is how many more problems can be shown to be equivalent to
\(d\)-dimensional convolution.

We developed a parameterized algorithm for multidimensional knapsack that generalized
techniques for one-dimensional knapsack.  Further, we used these techniques to
obtain a faster algorithm for Integer Linear Programming with upper bounds.

Finally, we gave an improved algorithm for Integer Linear Programming with upper
bounds that avoids the quadratic dependency on the dimension by increasing the
dependency on the number of variables.

\bibliography{lib}

\begin{thebibliography}{10}

\bibitem{axiotis19}
Kyriakos Axiotis and Christos Tzamos.
\newblock {Capacitated Dynamic Programming: Faster Knapsack and Graph
  Algorithms}.
\newblock In {\em 46th International Colloquium on Automata, Languages, and
  Programming (ICALP)}, volume 132, pages 19:1--19:13, Dagstuhl, Germany, 2019.
\newblock \href {https://doi.org/10.4230/LIPICS.ICALP.2019.19}
  {\path{doi:10.4230/LIPICS.ICALP.2019.19}}.

\bibitem{DBLP:journals/jcss/BlumFPRT73}
Manuel Blum, Robert~W. Floyd, Vaughan~R. Pratt, Ronald~L. Rivest, and
  Robert~Endre Tarjan.
\newblock Time bounds for selection.
\newblock {\em J. Comput. Syst. Sci.}, 7(4):448--461, 1973.
\newblock \href {https://doi.org/10.1016/S0022-0000(73)80033-9}
  {\path{doi:10.1016/S0022-0000(73)80033-9}}.

\bibitem{bremner}
David Bremner, Timothy~M. Chan, Erik~D. Demaine, Jeff Erickson, Ferran Hurtado,
  John Iacono, Stefan Langerman, and Perouz Taslakian.
\newblock Necklaces, convolutions, and x + y.
\newblock In {\em Algorithms -- ESA 2006}, pages 160--171, Berlin, Heidelberg,
  2006. Springer Berlin Heidelberg.
\newblock \href {https://doi.org/10.1007/11841036_17}
  {\path{doi:10.1007/11841036_17}}.

\bibitem{bringmannSubsetSum}
Karl Bringmann.
\newblock {\em A Near-Linear Pseudopolynomial Time Algorithm for Subset Sum},
  pages 1073--1084.
\newblock Society for Industrial and Applied Mathematics, 2017.
\newblock \href {https://doi.org/10.1137/1.9781611974782.69}
  {\path{doi:10.1137/1.9781611974782.69}}.

\bibitem{proxcel}
Marcel Celaya, Stefan Kuhlmann, Joseph Paat, and Robert Weismantel.
\newblock Improving the cook et al. proximity bound given integral valued
  constraints.
\newblock In {\em Integer Programming and Combinatorial Optimization - 23rd
  International Conference, {IPCO} 2022, Eindhoven, The Netherlands, June
  27-29, 2022, Proceedings}, volume 13265 of {\em Lecture Notes in Computer
  Science}, pages 84--97. Springer, 2022.
\newblock \href {https://doi.org/10.1007/978-3-031-06901-7_7}
  {\path{doi:10.1007/978-3-031-06901-7_7}}.

\bibitem{ChanMono}
Timothy~M. Chan and Moshe Lewenstein.
\newblock Clustered integer 3sum via additive combinatorics.
\newblock In {\em Proceedings of the Forty-Seventh Annual ACM Symposium on
  Theory of Computing}, STOC '15, page 31–40, New York, NY, USA, 2015.
  Association for Computing Machinery.
\newblock \href {https://doi.org/10.1145/2746539.2746568}
  {\path{doi:10.1145/2746539.2746568}}.

\bibitem{chen2023faster}
Lin Chen, Jiayi Lian, Yuchen Mao, and Guochuan Zhang.
\newblock Faster algorithms for bounded knapsack and bounded subset sum via
  fine-grained proximity results.
\newblock In {\em Proceedings of the 2024 {ACM-SIAM} Symposium on Discrete
  Algorithms, {SODA} 2024, Alexandria, VA, USA, January 7-10, 2024}, pages
  4828--4848. {SIAM}, 2024.
\newblock \href {https://doi.org/10.1137/1.9781611977912.171}
  {\path{doi:10.1137/1.9781611977912.171}}.

\bibitem{Chi}
Shucheng Chi, Ran Duan, Tianle Xie, and Tianyi Zhang.
\newblock Faster min-plus product for monotone instances.
\newblock In {\em Proceedings of the 54th Annual ACM SIGACT Symposium on Theory
  of Computing}, STOC 2022, page 1529–1542, New York, NY, USA, 2022.
  Association for Computing Machinery.
\newblock \href {https://doi.org/10.1145/3519935.3520057}
  {\path{doi:10.1145/3519935.3520057}}.

\bibitem{cygan}
Marek Cygan, Marcin Mucha, Karol W\k{e}grzycki, and Micha\l{} W\l{}odarczyk.
\newblock On problems equivalent to $(min,+)$-convolution.
\newblock {\em ACM Trans. Algorithms}, 15(1), jan 2019.
\newblock \href {https://doi.org/10.1145/3293465} {\path{doi:10.1145/3293465}}.

\bibitem{ddimknapsacklb}
Ilan Doron-Arad, Ariel Kulik, and Pasin Manurangsi.
\newblock Fine grained lower bounds for multidimensional knapsack, 2024.
\newblock \href {https://arxiv.org/abs/2407.10146} {\path{arXiv:2407.10146}}.

\bibitem{eisenbrand}
Friedrich Eisenbrand and Robert Weismantel.
\newblock Proximity results and faster algorithms for integer programming using
  the steinitz lemma.
\newblock {\em ACM Trans. Algorithms}, 16(1), nov 2019.
\newblock \href {https://doi.org/10.1145/3340322} {\path{doi:10.1145/3340322}}.

\bibitem{fullversion}
Kilian Grage, Klaus Jansen, and Björn Schumacher.
\newblock Convolution and knapsack in higher dimensions, 2024.
\newblock \href {https://arxiv.org/abs/2403.16117} {\path{arXiv:2403.16117}}.

\bibitem{gribanov}
Dmitry~V. Gribanov, Ivan~A. Shumilov, and Dmitriy~S. Malyshev.
\newblock Structured $(\min, +)$-convolution and its applications for the
  shortest/closest vector and nonlinear knapsack problems.
\newblock {\em Optim. Lett.}, 18(1):73--103, 2024.
\newblock \href {https://doi.org/10.1007/S11590-023-02017-5}
  {\path{doi:10.1007/S11590-023-02017-5}}.

\bibitem{Jin23}
Ce~Jin.
\newblock 0-1 knapsack in nearly quadratic time.
\newblock In {\em Proceedings of the 56th Annual {ACM} Symposium on Theory of
  Computing, {STOC} 2024, Vancouver, BC, Canada, June 24-28, 2024}, pages
  271--282. {ACM}, 2024.
\newblock \href {https://doi.org/10.1145/3618260.3649618}
  {\path{doi:10.1145/3618260.3649618}}.

\bibitem{Kronecker+1882+1+122}
L.~Kronecker.
\newblock Grundzüge einer arithmetischen theorie der algebraische grössen.
\newblock {\em Journal für die reine und angewandte Mathematik},
  1882(92):1--122, 1882.
\newblock \href {https://doi.org/doi:10.1515/crll.1882.92.1}
  {\path{doi:doi:10.1515/crll.1882.92.1}}.

\bibitem{Kunnemann}
Marvin K{\"{u}}nnemann, Ramamohan Paturi, and Stefan Schneider.
\newblock On the fine-grained complexity of one-dimensional dynamic
  programming.
\newblock In {\em 44th International Colloquium on Automata, Languages, and
  Programming, {ICALP} 2017, July 10-14, 2017, Warsaw, Poland}, volume~80 of
  {\em LIPIcs}, pages 21:1--21:15. Schloss Dagstuhl - Leibniz-Zentrum f{\"{u}}r
  Informatik, 2017.
\newblock \href {https://doi.org/10.4230/LIPICS.ICALP.2017.21}
  {\path{doi:10.4230/LIPICS.ICALP.2017.21}}.

\bibitem{proxspars}
Jon Lee, Joseph Paat, Ingo Stallknecht, and Luze Xu.
\newblock Improving proximity bounds using sparsity.
\newblock In {\em Combinatorial Optimization - 6th International Symposium,
  {ISCO} 2020, Montreal, QC, Canada, May 4-6, 2020, Revised Selected Papers},
  volume 12176 of {\em Lecture Notes in Computer Science}, pages 115--127.
  Springer, 2020.
\newblock \href {https://doi.org/10.1007/978-3-030-53262-8_10}
  {\path{doi:10.1007/978-3-030-53262-8_10}}.

\bibitem{DBLP:journals/jsc/Pan94}
Victor~Y. Pan.
\newblock Simple multivariate polynomial multiplication.
\newblock {\em J. Symb. Comput.}, 18(3):183--186, 1994.
\newblock \href {https://doi.org/10.1006/JSCO.1994.1042}
  {\path{doi:10.1006/JSCO.1994.1042}}.

\bibitem{polak21}
Adam Polak, Lars Rohwedder, and Karol W\k{e}grzycki.
\newblock {Knapsack and Subset Sum with Small Items}.
\newblock In {\em 48th International Colloquium on Automata, Languages, and
  Programming (ICALP 2021)}, volume 198 of {\em Leibniz International
  Proceedings in Informatics (LIPIcs)}, pages 106:1--106:19, Dagstuhl, Germany,
  2021. Schloss Dagstuhl -- Leibniz-Zentrum f{\"u}r Informatik.
\newblock \href {https://doi.org/10.4230/LIPIcs.ICALP.2021.106}
  {\path{doi:10.4230/LIPIcs.ICALP.2021.106}}.

\bibitem{larsdistinctalg}
Lars Rohwedder and Karol W\k{e}grzycki.
\newblock {Fine-Grained Equivalence for Problems Related to Integer Linear
  Programming}.
\newblock In {\em 16th Innovations in Theoretical Computer Science Conference
  (ITCS 2025)}, volume 325 of {\em Leibniz International Proceedings in
  Informatics (LIPIcs)}, pages 83:1--83:18, Dagstuhl, Germany, 2025. Schloss
  Dagstuhl -- Leibniz-Zentrum f{\"u}r Informatik.
\newblock \href {https://doi.org/10.4230/LIPIcs.ITCS.2025.83}
  {\path{doi:10.4230/LIPIcs.ITCS.2025.83}}.

\bibitem{williams}
Ryan Williams.
\newblock Faster all-pairs shortest paths via circuit complexity.
\newblock In {\em Proceedings of the Forty-Sixth Annual ACM Symposium on Theory
  of Computing}, STOC '14, page 664–673, New York, NY, USA, 2014. Association
  for Computing Machinery.
\newblock \href {https://doi.org/10.1145/2591796.2591811}
  {\path{doi:10.1145/2591796.2591811}}.

\end{thebibliography}

\appendix
\section{Omitted Proofs of Reductions of Section \ref{sec:redus}}\label{sec:omitted-redus}
\redbinenc*
\begin{proof}
We denote the $i$th item by the pair $(w^{(i)},p_{i})$ where $w^{(i)} \in \mathbb{N}^d$  is the weight vector and $p_i\in \mathbb{R}_{\ge 0}$ is the profit.
Consider an instance for the unbounded Knapsack problem and construct a 0/1-Knapsack instance as follows. For each item $(w^{(i)},p_{i})$ add items $(2^jw^{(i)},2^j p_{i})$ for all $j$ such that $2^j w^{(i)} \le t $.

Without loss of generality, we may assume that all $n$ weight vectors are pairwise different. If two weight vectors are the same, we would take the item with the higher profit. For each of these items in the unbounded instance, we then have at most $\lfloor \log (t_{max}) \rfloor +1$ items in the constructed bounded instance.

Assume now that in the unbounded instance some item $i$ is taken $k$ times in some solution. This choice can be replicated by taking the copies from $i$ that make up the binary representation of $k$. With this, we can convert any packing of either instance to a packing of the other. Therefore, both instances are equivalent.
\end{proof}


\begin{definition}
  For a $d$-dimensional array $A$ with size $L\in \mathbb{N}^d$, we say that
  $A$ is monotone increasing if for any positions $\nullvec{} \le v,u<L$ we have
  that $v\le u$ implies that $A_v \le A_u$.
\end{definition}
When we test for superadditivity, which is merely testing whether an array $A$
is an upperbound for the convolution of $A\oplus A$, it is helpful if we can work
with this simpler structure.

\begin{lemma}\label{lem:AddMon}
  For every $d$-dimensional array $A$ of size $L\in \mathbb{N}^d$, there is an array
  $A'$ with the same size that is non-negative and monotone increasing such that
  $A$ is superadditive iff $A'$ is superadditive. The values in this new array
  may increase by a factor of $\mathcal{O}(d L_{max})$.
\end{lemma}

\begin{proof}
Initially, set the entry at the zero vector position as $A'_{\nullvec{}}\coloneqq \max\{0, A_{\nullvec{}}\}$. Note that if $A_{\nullvec{}} > 0$ then $A$ cannot be superadditive since already $A_{\nullvec{}} + A_{\nullvec{}} > A_{\nullvec{}}$. Set $c \coloneqq 2 \max_{v\le L}{|A_v|}+1$ and with that define $A'_v \coloneqq A_v + c \cdot ||v||_1$ for all $v < L$.

We can immediately conclude that through the setting of $c$ we have that $A_v +c > 0$, hence $A'$ has no negative values. The monotony follows from the fact that $|A_v - A_{v+e^{(i)}}|<c $ for any unit vector $e^{(i)}$ and with that we have
\begin{equation*}
A'_{v} - A'_{v+e^{(i)}} = A_v +\norm{v}_1 c  - ( A_{v+e^{(i)}}+\norm{v+e^{(i)}}_1 c )= A_v - A_{v+e^{(i)}} - c < c- c = 0.
\end{equation*}
We now show the equivalence of superadditivity between the two arrays for any positions $u,v \neq \nullvec$ such that $u+v$ is a feasible position in the array $A$. Remember that the positions are non-negative, so $\norm{u+v}_1 = \norm{u}_1 + \norm{v}_1 $ and therefore:

\begin{align*}
A'_v + A'_u &\le A'_{v+u}  \\
 \Leftrightarrow A_v + c \cdot ||v||_1 + A_u + c \cdot ||u||_1 &\le A_{v+u} + c \cdot ||v+u||_1  \\
\Leftrightarrow A_v + A_u &\le A_{v+u}
\end{align*}
Consider now the case that one vector is actually the null-vector. Without loss of generality assume that $u = \nullvec$ and let $v$ be arbitrary, but such that $u+v$ is still a feasible position. We then have that $A'_u = 0$ if and only if $A_u \le 0$. The equality follows then immediately:

\begin{align*}
A'_v + A'_u &=A'_{v+u} = A'_{v}  \\
\Leftrightarrow A_v + A_u &\le A_{v+u} =A_{v}
\end{align*}
Let $L$ be the size of $A$ and remember that our input values through our
initial assumption are all in $[-W,W]$. Let $A'_p$ with $p= L-\constvec{1}$ be
the largest entry in $A'$. We note that the values of our new array may grow up
to
\begin{equation*}
A'_p = A_p + c \cdot \norm{p}_1 \le W + (2W+1) \cdot \norm{L-\constvec{1}}_1 \in
\mathcal{O}(||L||_{1}W) \leq \mathcal{O}(d L_{max} W). \qedhere
\end{equation*}
\end{proof}
Both properties, monotony and non-negative values, play an important role in the
reduction to \knapsunb, where we translate array positions to knapsack items.

\redprimdual*
\begin{proof}
  Let $A$ be an input array with size $L$ and assume thanks to
  \autoref{lem:AddMon} that $A$ is non-negative and monotone increasing. This
  might cause larger entries in the array that might incur additional overhead.
  We will account for this by adding a factor $\polylog(d L_{max})$ to the
  overall running time.

  Let $D \coloneqq \norm{L}_1 \cdot \max_{v< L}{A_v} + 1$. Set $t \coloneqq 2L$ as the new capacity
  for our knapsack instance. For every position $v< L$ in $A$, we define two
  items for our Knapsack instance. The first item, the primal item of position
  $v$, will have weight vector $w^{(v)} = v$ and profit $p_v = A_v$. The other
  item, the dual item of position $v$ will have weight vector
  $\overline w^{(v)} = t-v$ and profit $\overline p_v = D- A_v$.

  We can see that the profit $D$ is easily achievable for this Knapsack instance
  as any combination of primal and dual item for some position will yield a
  feasible packing that fills out the whole knapsack. We further argue that $D$
  is the maximum achievable profit if and only if $A$ is a superadditive array.

 First consider the case where $A$ is not superadditive, which means for some  $u,v< L$ we have that $A_v + A_u > A_{u+v}$. Combine now the primal items for positions $u,v$ with the dual item for position $u+v$. By definition of the weights the total weight of this packing is $t$. The profit of this packing then sums up to: $D-A_{v+u} + A_v + A_u > D$.

Consider now the case of $A$ being superadditive. We first argue that any knapsack solution may hold at most one dual item. Consider therefore any $1 \le i \le d$ and note that $(t-v)_i >2L_i -L_i =L_i$ for any position $v$. Therefore, adding two dual items would break all capacity constraints. Further, any optimal solution has to contain a dual item.

We may assume that no item has weight $\nullvec{}$, because these items can be removed and automatically added to the knapsack. This also means that there can be at most $||L||_1$ primal items in the knapsack and that the maximum profit of a packing containing only primal items is limited to $||L||_1 \cdot \max_{v< L}{A_v} < D$.

Let $S$ now be an optimal solution containing the dual item for position $v$. Assume further that $S$ contains multiple primal items and take two of these primal items for positions $u$ and $w$. Due to the superadditivity of $A$, we have that $A_u+A_w \le A_{u+w}$, so we can replace both of these primal items with the primal item of position $u+w$. This will not change the weight of the knapsack and only increase its profit, so the solution stays optimal. By doing this replacement over and over, we end up with an optimal packing that contains one dual item for position $v$ and only one primal item for some position $x$.

Note that we must have that $x\le v$ because otherwise the weight constraint would be broken and since we never changed the weight of the knapsack, this would have been true for the supposedly optimal solution $S$ as well. Since $A$ is monotone, we have that $A_x\le A_v$, so we can yet again replace the primal item for position $x$ with the primal item of position $v$ and end up with $A$ packing of profit $D$, which has to be optimal.
\end{proof}

\redblockmatrix*
\begin{proof}
  Let $A,B,C$ be the input matrices with the same size $L$. Let $L_1$ be the
  first entry of $L$. We will construct a new array $M$ of size $L'$ with
  $L'_1 \coloneqq L_1\cdot4$ and $L'_i = L_i$ for $1 < i \le d$, where we append each array in
  the first dimension. This means the previous three arrays appear next to
  each other preceded by a default block with negative values and $0$ in the
  $\nullvec$ position.  We formally define this array $M$ blockwise and address
  the four blocks of $M$ via index sets. For $0\le i \le 3$, denote with
  $P_i \coloneqq e^{(i)} \cdot i L_1$ the point where each array starts in the compound
  array (i.e. this will refer to the $\nullvec{}$ position in the original
  array) and denote the index set of each array as
  $I_i \coloneqq \{v < L' | P_i \le v < P_i + L\}$.

Let $K$ be large enough i.e. $K = 1+ 2\max_{v<L}\{|A_v|,|B_v|,|C_v|\}$, then define for $v < L'$:
\begin{figure}
\centering
\begin{tikzpicture}[scale=0.3]
\draw (0,0) rectangle (1,1) node[pos=.5] {$0$};
\draw (0,0) rectangle (6,6) node[pos=.5] {$-10K$};
\draw (6,0) rectangle (12,6) node[pos=.5] {$A+K$};
\draw (12,0) rectangle (18,6) node[pos=.5] {$B+4K$};
\draw (18,0) rectangle (24,6) node[pos=.5] {$C+5K$};

\node[below] at (0, 0) {$0$};
\draw[left] (0, 6) node {$L_2$};
\draw[below] (6, 0) node {$L_1$};
\draw[below] (12, 0) node {$2 L_1$};
\draw[below] (18, 0) node {$3 L_1$};
\draw[below] (24, 0) node {$4 L_1 -1$};
\end{tikzpicture}
\caption{Structure of the constructed array for $d=2$. Note that the annotation
  $X+s$ for an array $X$ and some value $s$, denotes the array where
  $(X+s)_v = X_v+s$ for every position $v$.}
\end{figure}
\[ M_v = \begin{cases}
        0 &\quad \text{ if } v\in I_0 \text{ and } v = \nullvec\\
				- 10K &\quad \text{ if } v\in I_0 \text{ and } v \neq \nullvec\\
				K + A_u &\quad \text{ if } v\in I_1 \text{ and } v = 1\cdot(L_1+1) + u\\
				4K + B_u &\quad \text{ if } v\in I_2 \text{ and } v = 2\cdot(L_1+1) + u\\
				5K + B_u &\quad \text{ if } v\in I_3 \text{ and } v  = 3\cdot(L_1+1) + u\\
     \end{cases}
\]
We now argue that $M$ being superadditive is equivalent to $C$ being an upper bound for the convolution of $A$ and $B$.
To see that, we make a case distinction over all $u,v < L'$ such that $u+v < L'$ and compare the respective entries in $M$. Note that we will skip symmetrical cases that appear by swapping the roles of $v$ and $u$.

Case 1: $v = \nullvec{}$.
We now have that $M_v + M_u = M_u = M_{u+v}$, hence the superadditivity property is always fulfilled.

Case 2: $v,u \in I_0\backslash\nullvec{}$.
In this case we get that $M_v + M_u = -20 K < -10K \le M_{u+v}$ for all $u$.

Case 3: $v \in I_0\backslash\nullvec{}$ and $u\in I_i$ for $i\ge 1$.
It follows now $M_v + M_u <  0 <  M_{u+v} $ since $u+v \in I_i$ for $i\ge 1$ and $M_v = -10$.

Case 4: $v, u\in I_1$.
Since two entries from $A$ added up will be smaller than $K$, we have that $M_v + M_u <  3K \le   M_{u+v} $ since $u+v \in I_i$ for $i\ge 2$.

Case 5: $v\in I_1$ and $u\in I_2$.  Note that $u+v \in I_3$ and set $v',u',w'$ such
that $M_v = K + A_{v'}$, $M_u = 4K + B_{u'}$ and $M_{u+v} = 5K +
C_{w'}$. Through definition we additionally have that $v' +u' = w'$.  We
conclude now that $M_v + M_u = K + A_{v'}+ 4K + B_{u'} = 5K + A_{v'}+ B_{u'}$
and with that we have
$M_v + M_u \le M_{u+v} \Leftrightarrow A_{v'}+ B_{u'} \le C_{w'}$ and hence in this case the
superadditivity property is fulfilled iff $C$ is an upper bound for the
convolution of $A$ and $B$.
\end{proof}

\begin{figure}
\centering
\begin{tikzpicture}[scale=0.6]

\node[above] at (4,3) {First Test};
\draw[fill=cyan] (0,0) rectangle (4,3);
\draw[] (0,0) rectangle (8,3);

\node[below] at (0,-0) {$0$};
\node[left] at (-0,3) {$L_2$};
\node[below] at (8,-0) {$L_1$};

\begin{scope}[xshift= -6cm, yshift=-4.5cm]

\node[above] at (4,3) {Second Test (First test returned false)};
\draw (0,0) rectangle (8,3);
\draw[fill=cyan] (0,0) rectangle (2,3);
\draw[pattern=north west lines] (4,0) rectangle (8,3);

\node[below] at (0,0) {$0$};
\node[left] at (0,3) {$L_2$};
\node[below] at (8,0) {$L_1$};

\end{scope}

\begin{scope}[xshift= +6cm, yshift=-4.5cm]

\node[above] at (4,3) {Second Test (First test returned true)};
\draw (0,0) rectangle (8,3);
\draw[fill=cyan] (0,0) rectangle (6,3);
\draw[pattern=north west lines] (0,0) rectangle (4,3);

\node[below] at (0,-0) {$0$};
\node[left] at (0,3) {$L_2$};
\node[below] at (8,-0) {$L_1$};

\end{scope}

\end{tikzpicture}
\caption{Sketch of the binary search. The colored area is the section from $A,B,C$ we apply the oracle on. The hatched part shows the part, that we do not specifically test for, because we know there is a faulty position in the unhatched part.} \label{fig:matrixbinarysearch}
\end{figure}

\begin{lemma}\label{lem:findUpperBoundConflict}
Consider a $T(\Pi(L),d)$ algorithm for \maxconvup. Let $A,B,C$ be three $d$-dimensional arrays of size $L$. We can then in time $\mathcal{O}(T(\Pi(L),d) \cdot d \cdot \log(L_{max}))$ find a position $v<L$ such that $C_v < (A \oplus B)_v$ or confirm that no such position exists.
\end{lemma}
\begin{proof}
Without loss of generality, we will assume that a desired position $v<L$ such that $C_v < (A \oplus B)_v$ exists. We can simply apply the algorithm for \maxconvup{} and if no such $v$ exists, then the algorithm will immediately return true and we are done.

We will argue that we can find $v$ in the desired time via multiple binary searches.
Fix a dimension $1\le i\le d$ and then do a binary search as follows in order to find an index $0\le u_i < L_i$ such that there is a position $v$ that we desire with $v_i = u_i$.

We first apply our algorithm to instances that have only half the amount of
entries i.e. that is halved in dimension $i$. To be more precise, we take the
sub-arrays from $A,B,C$ each of size $L - e^{(i)}\cdot \lfloor L_i /2\rfloor$. When the test
returns false, then there is a position in the first half of the array and we
can continue the binary search in there. Should the test return true, so for the
tested part we indeed have an upper bound property, then we need to look in the
half we did not test for the contradicting position. Note however that we still
need to include the already verified array parts to make sure all relevant
combinations are tested. See \autoref{fig:matrixbinarysearch} for an
illustration.

To find an actual position that violates the upper bound property, we use this binary search in the first dimension to find $u_1$. We know that there is a conflicting position $v$ with $v_1=u_1$. Before we apply the same approach in the second dimension however, we modify $C$ by setting any $C_p$ for $p<L$ with $p_1 \neq u_1$ to a sufficiently large value, i.e. $\infty$.

With this modification, we will make sure that no position $p<L$ with $p_1 \neq u_1$ will cause the \maxconvup{} test to fail. This may remove potential candidates for a conflicting position $v$ but as we found $u_1$, at least one position will remain. Now we can do the binary search for the second dimension and get $u_2$ with the conclusion that there is a conflicting position $v$ with $v_1=u_1$ and $v_2=u_2$.

We repeat this process, the binary search and respective modification of $C$,
for each dimension. By induction, we can conclude, that the resulting position
$u$ is in fact a conflicting position. Each application of the algorithm takes
time $T(\Pi(L),d)$. The algorithm is repeated $\mathcal{O}(\log(L_{max}))$ times for each
binary search and we use $d$ such binary searches leading to the overall running
time.
\end{proof}

\redsquarestuff*
\begin{proof}
  The proof consists of two parts. First, we will explain the general algorithm
  to find the correct convolution. The idea behind this procedure is to
  calculate every convolution entry via binary search. These binary searches are
  then processed in parallel, that is we start an array where each position
  holds a guess for that position. For the next iteration we test for each
  position whether its value is too low or too high and calculate for the next
  iteration a new full array with new guesses. In the second part, we will
  further specify how to achieve the desired running time.

  The general idea is that we will do binary search for the value of
  $(A\oplus B)_v$. Note that we may assume that every array only has non-negative
  entries, since we can add a large enough value to all positions and subtract
  it later from the convolution. Let $A_{max},B_{max}$ be the maximal entries in
  $A$ and $B$ respectively, then $K=A_{max}+B_{max}$ is an upper bound for the
  convolution value in any position and we do binary search for $C_v$ for all
  positions $v<L$ in the area from $0$ to $K$.

  We start by constructing an array $C^{(0)}$ with
  $C^{(0)}_v = \lfloor K/2 \rfloor$. We now want to identify for $C^{(0)}$ which entries are
  too small for the convolution. For that, we apply our oracle from
  \autoref{lem:findUpperBoundConflict} and identify a position $u$ that is
  contradicting the upper bound property.

  We know for this position $u$, that the current value is too small and we need
  to look for a larger value between $\lfloor K/2 \rfloor +1$ and $K$ in the next
  iteration. For the next iteration of the binary search, we will generate a new
  array $C^{(1)}$ and we can set $C^{(1)}_u$ to the next guess for the binary
  search. Before we continue however, we need to find all contradicting
  positions in $C^{(0)}$. To identify these, we update $C^{(0)}_u = K$ and apply
  our oracle again. Since $K$ is large enough, $C^{(0)}_u$ will not return as
  contradicting position and so our oracle will return a new contradicting
  position if there is one.

  By repeating this process, we can identify all values in our current array
  that are too small for the convolution and set them accordingly in
  $C^{(1)}$. At some point, our oracle will return that all entries fulfill the
  upper bound property. For all the entries that were never returned by our
  oracle, we know their value in the maximum convolution lies between $0$ and
  $\lfloor K/2 \rfloor$. We can therefore also set the appropriate next binary search guess
  in $C^{(1)}$ for these positions. We then can apply this same procedure for
  $C^{(1)}$ and inductively compute further arrays. After
  $\mathcal{O}(\log K)$ repetitions, we will have identified the convolution entries for
  all positions.

  Just blindly applying our oracle to the whole array however is too slow, so
  we will split our array in sub-arrays and apply our oracle in a specific way
  to be more efficient. Assume for the following part that every $L_i$ is a
  square number. We then split every array into $m = \sqrt{\Pi(L)} $ sub-arrays
  with up to $\sqrt{L_i}$ entries in each dimension $1\le i \le d$. Enumerate these
  blocks or sub-arrays from $1$ to $m$ and let $I_j$ denote the set of
  positions in block $1\le j \le m$. We now use our oracle to test certain
  combinations of these sub-arrays that we call chunks.

  We can test for two chunks $i,j$ whether for every $v \in I_i$, $u \in I_j$ we
  have that $A_v + B_u \le C_{v+u}$ or find positions that violate this constraint
  using the oracle. The combination of $v+u$ however may be contained not in
  only one but $2^d$ many chunks. We therefore apply the algorithm to arrays
  of size $2\sqrt L_i$ in each dimension. To be precise, we define $A',B'$ to be
  the arrays defined by the chunks $I_i$ and $I_j$ and expand these to the
  desired size by filling them up with small dummy entries $-K$. $C'$ is then
  the sub-array that contains all positions $v+u$ for $v \in I_i$ and $u \in I_j$.

  Now we can apply our oracle in the sense of the binary search procedure that
  we described above. Assume we already have one array $C^{(i)}$ and combine
  every possible combination of chunks. One run of the oracle based on
  \autoref{lem:findUpperBoundConflict} takes time
  $\mathcal{O}(T(2^d\sqrt{\Pi(L)},d) \cdot d \cdot \log(\sqrt{L_{max}}))$, since we expanded the
  size of each chunk to gain equal sized arrays. We repeat this test for all
  possible combinations of chunks, which are in total $m^2 = \Pi(L)$.

  Further, if a chunk yields a contradicting position, we update the value in
  $C^{(i+1)}$ and we repeat the test of the chunk again. However, by setting the
  faulty position with a large value each position can only contradict the upper
  bound property once and therefore the necessary number of repetitions for this
  is again bounded by $\Pi(L)$. In total, we can upper bound the running time by
  $\mathcal{O}(\Pi(L) \cdot T(2^d\sqrt{\Pi(L)},d) \cdot d \cdot \log(\sqrt{L_{max}}))$. Remember that we
  omitted the factor of $\mathcal{O}(\log K)$ for the overall binary search on each
  position. Since by definition of $K$ we have $K \le 2W$ the desired result
  follows.

  Note that if $L_i$ is not a square number, we might end up with additional
  incomplete sub-arrays that we need to combine. This might incur another factor
  of $2^d$ for the number of arrays that have to be combined. As $d$ is fixed,
  this will also not alter the running time further.
\end{proof}

\subsection{Reduction from Knapsack to Convolution}\label{section:bringmann}

In this section, we present a reduction from \knapsone{} to \maxconv{} and
therefore proof \autoref{theo:red:bringmann}. This reduction will be given
through an algorithm that solves \knapsone{} by computing multiple
convolutions. In fact, we will solve the proposed Knapsack problems for all
possible capacities $t'\le t$.

The solution we give is then an array $A$, call it \textit{knapsack array} in
the following, of size $t+\constvec{1}$ such that $A_{t'}$ denotes the maximum
profit of a knapsack with capacity $t'\le t$. Our algorithm follows a basic idea:
We randomly split the items from the instance into different subsets, compute
solution arrays for these subsets and then combine them using convolution.

We note that these arrays have size $t+\constvec{1}$ and that combining arrays of this size takes time $T(\Pi(t+\constvec{1}),d) \in \mathcal O(2^d T(\Pi(t),d))$. When $d$ is fixed, the running time would be $\mathcal O(T(\Pi(t),d))$. For this reason, we will omit that the knapsack array we compute is technically slightly larger than $t$ and still write that computing convolutions of these sizes takes time $\mathcal{O}(T(\Pi(t),d))$.

The randomized algorithm we present for this reduction is a simple variation of the one introduced by Cygan et al.~\cite{cygan} and Bringmann \cite{bringmannSubsetSum}. The trick is to choose specific subsets of items, solve the Knapsack problem for these sizes and then combine the solutions via convolution. As such, it is important to decide how to build fitting sub-instances in the higher dimensional case.

In the following we will present these algorithms and prove their correctness. We start off with a randomized algorithm that will find an optimal solution with a bounded number of items via a technique called color-coding \cite{bringmannSubsetSum}. (see also \autoref{lst:colorcoding}).

\begin{lemma}\label{lem:colcod}
  For any Knapsack instance with item set $I$ and capacity $t$, denote with
  $\Omega(I)_v$ the maximum profit of a subset $I'\subset I$ with total weight smaller or
  equal to $v\le t$.

  Let $\delta \in (0,1)$ and $k \in \mathbb{N}$.  There exists an algorithm that for every
  Knapsack item set $I$ and capacity $t$ computes a random array $S$ in time
  $\mathcal{O}(T(\Pi(t),d)k^2\log{(1/\delta)})$ such that for any $Y\subseteq I$ with
  $|Y| \le k$ and every capacity $v\le t$ we have
  $\Omega(Y)_v \le S_v \le \Omega(I)_v$ with probability $\ge 1-\delta$.
\end{lemma}
\begin{proof}
  We initialise $k^2$ arrays $Z^{(1)}, \cdots , Z^{(k^2)}$ of size
  $t+\constvec{1}$ and randomly distribute items from $I$ into
  $Z^{(1)},\cdots, Z^{(k^2)}$. The probability for an item to land in either array
  is equally distributed. When an item with weight $v$ and profit $p$ is
  assigned to $Z^{(i)}$, we set $Z^{(i)}_v \eqqcolon p$. If two items of the same
  weight are assigned to one array, we keep the highest profit. Ultimately we
  want that items from an optimal solution are split among our arrays.

Take an optimal solution $Y$ with at most $k$ elements. The probability that all these items are split in different $Z^{(i)}$ can be bounded by:

\begin{equation*}
  \frac{k^2-1}{k^2} \cdot \frac{k^2-2}{k^2}  \cdots  \frac{k^2-(|Y|-1)}{k^2} \ge
  \paren[\Big]{1- \frac{k-1}{k^2}}^k \ge \paren[\Big]{1- \frac{1}{k}}^k \ge \paren[\Big]{\frac{1}{2}}^2
\end{equation*}
So with probability of at least $1/4$, all elements of $Y$ are split among the
$k^2$ arrays, that is no two items in $Y$ are in the same array. We build the
convolution of all $Z^{(i)}$'s and get $S$ as a result in running time
$\mathcal{O}(T(\Pi(t),d)k^2)$. By repeating this process
$\mathcal O(\log{(1/\delta)})$ times, we can guarantee the desired probability, while
achieving the proposed total running time.
\end{proof}

\lstset{
    escapeinside={(*}{*)},          
}

\begin{lstlisting}[caption={ColorCoding}\label{lst:colorcoding},float=h,
abovecaptionskip=-\medskipamount,]
ColorCoding((*$ I,t,k,\delta$*))
for (*$j \in \{1, \cdots ,\lceil \log_{4/3}(1/\delta)\rceil\}$*)
    randomly partition (*$I = Z^{(1)} \cup \cdots \cup Z^{(k^2)}$*)
    compute (*$P^{(j)} = Z^{(1)} \oplus \cdots \oplus Z^{(k^2)}$*)
end for
for all (*$v\le t$*) set (*$S_v \coloneqq \max_j P^{(j)}_v$*)
return (*$S$*)
\end{lstlisting}

We now split a given item set $I$ with $n=|I|$ into disjoint layers based on the capacity of the knapsack ${t}$. Denote with $L_j^i$ the set of items with a weight vector $w$ such that
$w_i \in ({t}_i/2^j,{t}_i/2^{j-1}]$ for $i<d$ and $j< \lceil \log n\rceil$. Denote with $L^i_{\lceil \log n\rceil}$ the last layer holding items whose weight vector has that $w_i \le {t}_i / 2^{\lceil \log n\rceil -1}$.

Now every item appears in $d$ different layers. Choose therefore for an item one layer $L^i_j$ that contains it and such that $j$ is minimal. If there are multiple possibilities, choose arbitrarily among these and remove said item from any other layers. We note that from a layer $L_j^i$, we can choose at most $2^{j-1}$ items for a knapsack solution as otherwise the $i$th capacity constraint would be exceeded. By choice of layers, it is also possible for any solution to contain that many items. As we have chosen $j$ minimal, there can be no other weight constraint invalidating our solution.

In the following lemma, we will now prove how to solve one of these layers. The algorithm for that can be seen in \autoref{lst:layer}.

\begin{lstlisting}[caption={Algorithm to solve one layer using color-coding}\label{lst:layer},float=h,
abovecaptionskip=-\medskipamount,]
SolveLayer((*$L,t,j,\delta$*))
Set (*$l=2^j, \gamma= 6\log(l/\delta)$*)
Let (*$m= l/\log(l/\delta)$*) be rounded to the next power of 2.
if (*$l < \log(l/\delta)$*) then return ColorCoding((*$L,t,l,\delta$*))
randomly partition (*$L = A_1 \cup \cdots \cup A_m$*)
for (*$k \in [m]$*)
    (*$P^{(j)}$*) = ColorCoding((*$A_k,(2\gamma/l)t, \gamma, \delta/l$*))
for (*$h \in [\log m]$*)
    for (*$k \in m/(2^h)$*)
        (*$P_k = P_{2k-1} \oplus P_{2k}$*)
    end for
end for
return (*$P_1$*)
\end{lstlisting}

\begin{lemma}\label{lem:solveLayer}
  For all layers $L_j^i$ and $\delta \in (0, 1/4]$ there exists an algorithm that
  computes an array $W$ in time
  $\mathcal{O}( T(12 \Pi(t),d)\log (d\cdot t_{max})\log^3{(2^j/\delta)})$, where for each capacity
  $v\le t$ we have that $W_v=\Omega(L_j^i)_v$ with a probability of at least $1-\delta$.
\end{lemma}

\begin{proof}

Consider a layer $L^i_j$ and set the following parameters based on the algorithm in \autoref{lst:layer}: $l=2^j$, $m=l/\log{(l/\delta)}$ rounded up to the next power of $2$ and $\gamma = 6\log{(l/\delta)}$.

Consider the case where $l<\log{(l/\delta)}$, we then compute the desired result in time $\mathcal{O}(T(\Pi(t),d)l^2\log{(1/\delta)}) \in \mathcal O(T(\Pi(t),d)\log^3{(l/\delta)})$. We do this by applying \autoref{lem:colcod} with probability parameter $\delta$ and bound on the number of items $l$.

In the general case where $l\ge \log{(l/\delta)}$, we split the item set into
$m$ disjoint subsets $A_1,\cdots , A_m$. We apply again \autoref{lem:colcod}.  We use
$\delta/l$ for the probability parameter and allow $\gamma$ items for each subset.
Furthermore, we set the capacity to $\frac {2\gamma}{l}\cdot t$, which is the maximum
weight of any packing that uses at most $\gamma$ items with size at most
$2^{j-1}$, which is the maximum size in layer $L^{i}_{j}$. To calculate a
solution array of $A_k$ via algorithm \autoref{lst:colorcoding}, we need time:
\begin{equation*}
\mathcal{O}(T( \Pi(2\gamma/l \cdot t) ,d)\gamma^2\log{(l/\delta)})
\in \mathcal{O}(T(\Pi(12\log{(l/\delta)}/l \cdot t),d)\log^3{(l/\delta)}).
\end{equation*}
We want to note two things for our estimation. First, for $x > 0$ we have that $T(x,d) \in \Omega(x)$, so for $y \ge 1$ we have that $y \cdot T(x,d) \le (T(y\cdot x,d)$. For the same reason, we have $y\cdot \Pi(v) \le \Pi(y\cdot v) $ for some  vector $v$ with $v \ge \constvec{1}$.
Denote with $m' = l/\log{(l/\delta)}$ the value of $m$ before rounding and conclude that $2m' \ge m$.

Together to solve all $A_k$ we therefore require in total:
\begin{align*}
\mathcal{O}( m' T(\Pi(12\log{(l/\delta)}/l \cdot t),d)\log^3{(l/\delta)})
&\in \mathcal{O}(T(\Pi(m'\cdot12\log{(l/\delta)}/l \cdot t),d)\log^3{(l/\delta)})\\
&\in \mathcal{O}(  T(\Pi(12 \cdot t),d)\log^3{(l/\delta)})
\end{align*}
 We now need to combine all the $A_k$ and we do so in a binary tree fashion. In round $h \in [\log m]$, we have $m/{2^h}$ convolutions, each with a result array of size $2\cdot 2^h \gamma/l \cdot t$. The convolutions take total time:

 \begin{align*}
 \sum_{h=1}^{\log m}{\frac{m}{2^h} T(\Pi(2 \cdot 2^h \gamma /l \cdot t),d)}
 &\le \sum_{h=1}^{\log m}{2\cdot \frac{m'}{2^h} T(\Pi(2^h \cdot  12 \log{(l/\delta)}/l \cdot t),d)} \\
 &\le \sum_{h=1}^{\log m}{2 \cdot T(\Pi(\frac{m'}{2^h} \cdot 2^h \cdot  12 \log{(l/\delta)}/l \cdot t),d)} \\
 &= \sum_{h=1}^{\log m}{2\cdot T(\Pi( 12 \cdot t),d)} \in \mathcal{O}(T(\Pi(12t),d) \log m)
 \end{align*}

  We might require arithmetic overhead due to the large values of profits. A
  knapsack with size $t$ may only hold $||t||_1$ items whose weight is not
  $\nullvec{}$. Note that we may remove items with weight $\nullvec{}$ without
  loss of generality and we can conclude that the maximum profit is bounded by
  $||t||_1 \cdot W \in \mathcal O (d \cdot t_{max} W)$. We add a factor of $\log(d \cdot
  t_{max})$ to the running time to compensate for the arithmetic overhead.

  All that is left is to argue that our algorithm works correctly and with the proposed probability. We note that the first case follows immediately from the \autoref{lem:colcod}. For the general case fix an optimal item set $Y$ and set $Y_k = Y \cap A_k$ for $k\in [m]$. Based on the distribution of items $|Y_k|$ is distributed as the sum of $|Y|$ independent Bernoulli random variables. Each variable has a success probability of $1/m$ and the expected value is $\mu \coloneqq \mathbb{E}[|Y_k|]= |Y|/m \le l/m \le \log(l/\delta)$. According to a standard Chernoff bound, it holds that $Pr[|Y_k| \ge \lambda] \le 2^{-\lambda}$ for any $\lambda \ge 2e\cdot \mu$. By using $\gamma = 6\mu\ge 2e \cdot \mu$, we get that
  $Pr[|Y_k| \ge \gamma] \le 2^{-\gamma}= \frac{1}{2^\gamma} \le \frac{1}{2^{\log(l/\delta)}}
  =  \frac{\delta}{l}$ and this describes the probability that some partition $A_k$ has too many elements from $Y$.

  With a probability of each at least $1-\frac{\delta}{l}$, we have that $|Y_k| \le \gamma$ and therefore our algorithm for \autoref{lem:colcod} is applicable to find the solution generated by $Y_k$ or one of equal value for the same weight. The probability that this happens for all partitions is at least $(1-\frac{\delta}{l})^m \ge 1-m\frac{\delta}{l}$ using the Bernoulli inequality. The probability that we find the solution generated by $Y_k$ based on \autoref{lem:colcod} is also at least $1-\frac{\delta}{l}$, so the probability that this works also for all $i\in [m]$ is $\ge 1-m \frac{\delta}{l}$. The probability of both events, the proper split and finding the right solution is then at least $(1-m \frac{\delta}{l})^2 \ge 1-2\cdot m\frac{\delta}{l}$. We can further note that $m\le l/2$ and therefore $ 1-2\cdot m\frac{\delta}{l} \ge 1-\delta$.
\end{proof}
Now that we can solve each layer, all that remains is to combine the solutions for each layer. To do so, we simply apply convolution again.

\begin{proof}[Proof of \autoref{theo:red:bringmann}]

As discussed earlier, we split the set of items into the respective layers and solve each layer via \autoref{lem:solveLayer}. We use as probability parameter $\delta' \coloneqq \delta/\lceil \log(n) \rceil$. To solve a layer, we require time
$\mathcal{O}( T(12 \Pi(t),d)\log (d\cdot t_{max})\log^3{(\log n/\delta)})$ and we do this for all $d\cdot \lceil \log n \rceil$ layers. We then use the convolution algorithm to combine the resulting arrays from each layer in time $\mathcal O(T(\Pi(t), d)\cdot  d\cdot \log n)$. Putting all this together, we require a total time in $\mathcal{O}(T(12 \Pi(t),d)\log (d\cdot t_{max})\log^3{(\log n/\delta)}\cdot  d\cdot \log n)$.

For the probability, consider some optimal solution $Y$ and $Y_j^i \coloneqq L_j^i \cap Y$. Based on \autoref{lem:solveLayer}, we find $Y_j^i$ or a solution of equal profit for the same weight with a probability of $\ge 1-\delta' = 1-\delta/\lceil \log(n) \rceil $. The chance that we find the proper solutions over all layers is then at least $(1-\delta/\lceil \log(n) \rceil)^{\lceil \log n \rceil} \ge 1-\delta $ with the Bernoulli inequality.
\end{proof}

\section{Calculate \texorpdfstring{\(d\)}{d}-dimensional Convolution with 1-dimensional Convolution}\label{sec:calc-d-conv-helix}
We can interpret \((\max, +)\)-convolution as multiplication of polynomials over
the \((\max, +)\)-semiring, also called the Arctic semiring.  \(d\)-dimensional
convolution is multiplication of polynomials with \(d\) variables
\(x_1, \ldots, x_d\).  More precisely, let \(A, B\) be two \(d\)-dimensional arrays
of size \(L\).  Then, the corresponding polynomial to \(A\) is given by
\begin{equation*}
  \label{eq:poly-trans}
  p(A) \coloneqq \sum_{0\leq v < L} A[v]\prod_{i=1}^dx_i^{v_i}.
\end{equation*}
The Kronecker map \cite{Kronecker+1882+1+122} as described in
\cite{DBLP:journals/jsc/Pan94} maps \(x_{k + 1}\) to \(x^{D(1)\cdots D(k)}\) for
every \(k\in[d - 1]_0\) where \(D(i) = 2L_i - 1\) for every \(i\in[d]\).  This map
is generalized to polynomials accordingly and allows us to obtain an
one-dimensional array this way by adding dummy values.  An example for this
construction is depicted in \Cref{fig:ex:poly-trans}.  Let \(\tilde{p}(A)\) and
\(\tilde{p}(B)\) be the results of this transformation.

\begin{figure}[bth]
  \begin{center}
  \begin{tikzpicture}[semithick, 3d view]
    \foreach \z in {0, 1} {
      \begin{scope}[canvas is xz plane at y=-2.5*\z]
        \draw[color=gray, step=.5] (0, 0) grid +(1.5, -2);
        \foreach \y in {0, 1, 2, 3} {
          \foreach \x in {0, 1, 2} {
            \node[transform shape] at (.5*\x + .25, -.5*\y - .25)
            {\pgfmathparse{int(65 + (1 - \z) * 12 + \y * 3 + \x)}\char\pgfmathresult};
          }
        }
      \end{scope}
    }
  \end{tikzpicture}

  \vspace*{.5cm}

  \begin{tikzpicture}[semithick, scale=.8, transform shape]
    \foreach \x in {0, ..., 52} {
      \tikzmath{
        int \a, \b;
        \a = mod(\x, 27);
        \b = \x / 27;
        if \b > 0 then {
          \a = \a + 1;
        };
      }
      \node[draw, inner sep=0pt, minimum width=.5cm, minimum height=.5cm] at (.5*\a, -\b) {\tikzmath{
        int \a, \b, \c;
        \a = mod(\x, 5);
        \b = mod(\x / 5, 7);
        \c = mod(\x / 35, 3);
        if \a < 3 && \b < 4 && \c < 2 then {
          int \res;
          \res = 65 + 12 * \c + 3 * \b + \a;
          print{\char\res};
        } else {
          print{\(-\)};
        };
      }};
    }
    \draw[-{Stealth[round]}] (13.5, 0) .. controls +(.5, 0) and +(.5, 0) .. (13.5, -.5) -- (0, -.5) .. controls +(-.5, 0) and +(-.5, 0) .. (0, -1);
  \end{tikzpicture}
  \end{center}
  \caption{Example for the transformation in 3 dimensions with
    \(L = (3, 4, 2)\tran\).}
  \label{fig:ex:poly-trans}
\end{figure}

Now, we can calculate \(\tilde p(A) \cdot \tilde p(B)\) with an one-dimensional
\((\max, +)\)-convolution algorithm and transform the result back according to
the transformations above to obtain an \(d\)-dimensional array.  The result is
correct since the polynomial multiplication is correct due to the added padding
in the definition of \(D(i)\) \cite{DBLP:journals/jsc/Pan94}.

To show \Cref{prop:calc-conv-helix} is remains to show that \(\deg(\tilde p(A))
+ 1 \leq
2^d\prod(L)\) since \(\deg(\tilde p(A))\) is the length of the resulting arrays for
the one-dimensional convolution.  We have
\begin{align*}
  2\deg(\tilde p(A))\tag{\(L - \constvec{1}\) is the maximum index.}
  &= 2\deg\paren[\Big]{A[L-\constvec{1}]\prod_{i=1}^d\paren{x^{D(1)\cdots D(i - 1)}}^{L_i - 1}}\\
  &= 2\sum_{\ell = 1}^{d}(L_\ell - 1)\prod_{m=1}^{\ell - 1}D(m)\\
  &\leq \sum_{\ell = 1}^{d}\paren[\Big]{\prod_{m=1}^{\ell - 1}2 L_m}2L_{\ell}\\
  &= 2L_1 + 2L_1 \sum_{\ell = 2}^{d}\paren[\Big]{\prod_{m=1}^{\ell - 1}2L_m}2L_{\ell}\\\tag{With
  induction}
  &= \sum_{\ell = 1}^{d}2^{\ell}\prod_{m = 1}^{\ell}L_m\\
  &= \paren[\bigg]{\sum_{\ell = 1}^{d}2^{\ell}}\prod(L)\\
  &\leq 2^{d+1}\prod(L)\\
  \implies& \deg(\tilde p(A)) \leq 2^d\prod(L)
\end{align*}
In conclusion we have show the following lemma.
\calcconvhelix*

\end{document}